\documentclass[a4paper,fleqn]{cas-sc}

\usepackage[numbers]{natbib}

\def\tsc#1{\csdef{#1}{\textsc{\lowercase{#1}}\xspace}}
\tsc{WGM}
\tsc{QE}

\usepackage{amsmath,amsfonts,amssymb}
\usepackage{algorithmic}
\usepackage{algorithm}
\usepackage{array}
\usepackage[caption=false,font=normalsize,labelfont=sf,textfont=sf]{subfig}
\usepackage{textcomp}
\usepackage{stfloats}
\usepackage{url}
\usepackage{verbatim}
\usepackage[dvipsnames]{xcolor}
\usepackage{tikz}
\usetikzlibrary{calc, patterns, positioning, circuits.ee.IEC, decorations.pathmorphing}
\usepackage{booktabs}
\usepackage{float}
\usepackage{comment}
\usepackage{dsfont} 
\usepackage[siunitx,RPvoltages]{circuitikz} 

\usepackage{xcolor}

\usepackage{amsthm}

\newtheorem{lemma}{Lemma}
\newtheorem{assumption}{Assumption}

\newtheorem{remark}{Remark}

\newtheorem{theorem}{Theorem}

\newtheoremstyle{plain}
  {} 
  {} 
  {\itshape} 
  {} 
  {\bfseries} 
  {.} 
  {.5em} 
  {} 
\theoremstyle{plain}

\DeclareMathOperator{\col}{col} 

\begin{document}
\let\WriteBookmarks\relax
\def\floatpagepagefraction{1}
\def\textpagefraction{.001}

\shorttitle{Data-Driven Dead-Zone Compensation}    

\shortauthors{M. Mazare and H. Ramezani}  

\title [mode = title]{Data-Driven Dead-Zone Compensation via Projection in Predictive Control Setting}  

\author[1]{Mahmood Mazare}


\ead{mazare@sdu.dk}

\author[1]{Hossein Ramezani}

\cormark[1]

\ead{ramezani@sdu.dk}

\affiliation[1]{organization={Institute of Mechanical and Electrical Engineering, University of Southern Denmark},
            city={Odense},
            country={Denmark}}

\cortext[1]{Corresponding author: Hossein Ramezani}

\begin{abstract}
Actuator dead-zones are a common and troublesome nonlinearity in motion control: a band of commanded effort over which the plant does not respond, leaving a steady-state offset or a limit cycle. This paper proposes a data-driven architecture that compensates such mismatches without a model of the plant and without any parameterization of the dead-zone. The central idea is to identify, alongside the velocity-form predictor used for control, a second \emph{absolute} subspace predictor. Because the absolute predictor carries no integral action, it behaves as a data-driven steady-state sensor, so a persistent actuator mismatch appears as a proportional prediction residual. Embedding this residual as a proxy in a behavioral Hankel matrix reduces the mismatch estimate to a single fixed orthogonal projection evaluated online, with no dynamic estimator, no injected probing signal, and no run-time prediction-error computation. Integrated into a subspace predictive controller, the framework is shown to be recursively feasible and practically input-to-state stable, and it recovers offset-free tracking once the dead-band traversal settles. The approach is validated in real time on a sixth-order, lightly damped Quanser multi-DOF torsion system, whose complex-conjugate poles give a lightly damped open-loop response, achieving offset-free tracking across a $\pm 0.18$\,V actuator dead-band. A second study on a high-precision power amplifier shows that the same architecture rejects dead-time-induced nonlinearities in fast-switching power electronics.
\end{abstract}


\maketitle

\section{Introduction}
Actuator dead-zones are a persistent source of performance loss in motion control \cite{Liuetal2021}. A dead-zone is a band of commanded effort over which the plant produces no response \cite{Zhangetal2020}. Whether it originates from static friction in geared transmissions or from voltage thresholding in servo drives, the effect on the closed loop is similar \cite{Shouetal2020}: the system approaches the reference and then stalls, leaving a steady-state offset \cite{Lietal2021}. In lightly damped, flexible mechanisms the problem is more acute, because the residual control effort needed to cross the dead-band can excite unmodeled oscillatory modes \cite{Wangetal2021}.

Classical remedies typically apply an explicit model-based inverse to pre-process the controller output \cite{Zhangetal2023}. Accurate parameterization can restore a near-linear input path, but an exact inverse is discontinuous \cite{Yuetal2024}, and this discontinuity can induce chattering \cite{Caoetal2024} and often limits the loop to ultimately bounded rather than exact tracking \cite{Suetal2018}. To mitigate this, smooth inverse models \cite{Bessaetal2010} and adaptive update laws \cite{Bessaetal2022} have been proposed. A related route avoids explicit inversion by decomposing the dead-zone into a nominal linear gain plus a bounded, disturbance-like term \cite{Bessa2022}. This decomposition combines well with nonlinear stabilizing techniques such as sliding-mode control \cite{Jungetal2022} and dynamic surface control \cite{ChanZhengetal2022}, and with passivity-based designs \cite{Fengetal2025}. These explicit compensators, however, depend on parametric identification \cite{Cequeiraetal2010}, so the loop can degrade under asymmetric wear \cite{Sunaretal2021} or time-varying faults \cite{Wangetal2023}.

When the dead-zone is unknown or coupled with other hard nonlinearities, learning-based approximation methods offer an alternative. Neural networks have been trained to learn the dead-zone map directly from operating data \cite{Jiaetal2019}, with architectures targeting robotic manipulators \cite{Zhouetal2019} and time-varying delays \cite{Maetal2022}. Fuzzy logic systems have been used to approximate inverse models and uncertainty bounds \cite{Jang2019}, including for uncertain time-delay systems \cite{Chiang2013}, and this line of work has developed steadily in recent years \cite{Tangetal2024}. Since exact cancellation would require switching that real actuators cannot achieve, such learners are usually embedded within a stabilizing baseline loop \cite{Cheongetal2013}; combining them with robust tracking gives a capable hybrid architecture \cite{Xiangetal2017}. Other frameworks, such as Active Disturbance Rejection Control (ADRC), avoid explicit modeling by estimating the dead-zone's net effect as a lumped total disturbance \cite{Wangetal2020}, shifting the burden from parameter learning to disturbance estimation \cite{Ahietal2018}. These observer-based approaches are flexible, but they add tuning burden and residual approximation error, and they still assume a nominal plant structure or rely on fast dynamic observers that amplify measurement noise \cite{Ebrahimietal2024}.

When a plant model is unavailable, behavioral data-driven predictive control provides a model-free alternative. Building on Willems' fundamental lemma \cite{willems2005}, data-enabled predictive control (DeePC) replaces identified state-space models with raw input-output trajectories \cite{coulson2019}, and subspace predictive control (SPC) offers a closely related, computationally efficient formulation \cite{favoreel1999}. Recent work has extended these formulations to improve robustness \cite{depersis2020}. Cast in velocity form, these controllers inherit offset-free tracking at no additional computational cost \cite{Lazaroffset}. This convenience carries a trade-off, however: the integral action that provides offset-free tracking also removes the steady-state signature of the dead-zone. The incremental controller can integrate its way around small mismatches, but it does not observe the actuator drop-off it is compensating, and the resulting degraded transients are exactly the behavior one set out to avoid.

This paper addresses this trade-off with a data-driven, projection-based scheme that requires no plant model, no parameterization of the dead-zone, and no learning-based tuning. The key observation is that the lost steady-state information can be recovered by identifying a second \emph{absolute} predictor alongside the incremental controller. Because the absolute predictor has no integral action, it acts as a data-driven DC sensor, and a persistent actuator mismatch appears as a proportional steady-state residual. Embedding this residual as a proxy in a behavioral Hankel matrix reduces the mismatch estimate to a single fixed orthogonal projection, evaluated online as one matrix-vector product without a dynamic estimator, an injected probing signal, or run-time prediction-error computation. Integrated into a subspace predictive controller, the estimate lets the loop shift its baseline effort to cross the dead-band and recover offset-free tracking.

The main contributions of this paper are summarized as follows:
\begin{itemize}
    \item We show that the integral action of velocity-form data-driven predictors removes the steady-state information needed to identify constant actuator mismatches, and we recover this information with a companion \emph{absolute} subspace predictor that acts as a data-driven DC sensor.
    \item We show that the steady-state absolute prediction error maps the physical mismatch through a fixed, data-identified gain, and we embed the resulting proxy in a behavioral Hankel matrix so that the mismatch is extracted online via a single, offline-computed orthogonal projection.
    \item We establish bounded estimation, recursive feasibility, and practical input-to-state stability for the integrated estimator-controller loop, and we show that offset-free tracking is recovered once the dead-band traversal settles.
    \item We validate the architecture in real time on a sixth-order, lightly damped multi-DOF torsion system subject to an actuator dead-zone, and in simulation on an industrial high-precision power amplifier.
\end{itemize}

The remainder of this paper is organized as follows. Section II formulates the dead-zone compensation problem in a model-free setting. Section III identifies the companion absolute and incremental subspace predictors. Section IV derives the absolute prediction error as an exact mismatch proxy. Section V develops the fixed geometric projection that extracts the mismatch online. Section VI establishes recursive feasibility and practical stability. Section VII reports the real-time experimental results, and Section VIII concludes the paper.

\section{Problem Formulation and Control Objective}\label{sec:problem}

Consider a discrete-time, linear time-invariant (LTI) system, whose internal dynamics are represented by the following state-space model:
\begin{subequations}
\begin{align}
    x(k+1) &= A x(k) + B u_{\text{eff}}(k), \label{eq:state} \\
    y(k)   &= C x(k), \label{eq:output}
\end{align}
\end{subequations}
where $x(k) \in \mathbb{R}^n$ is the unmeasurable internal state vector, and $y(k) \in \mathbb{R}^p$ is the measured output vector. The system matrices $A$, $B$, and $C$ are assumed to be strictly unknown, necessitating a purely data-driven framework for both estimation and control.

\subsection{Actuator Mismatch and the Dead-zone Challenge}
The variable $u_{\text{eff}}(k) \in \mathbb{R}^m$ denotes the actual effective input physically acting upon the plant dynamics. In practical systems such as multi-torsion mechanical setups, actuators are not ideal; they are subject to nonlinear physical mismatches, most notably dead-zones and Coulomb friction. Mathematically, a dead-zone implies that the effective input is degraded by an unknown nonlinear loss function:
\begin{equation}
    u_{\text{eff}}(k) = u_{\text{cmd}}(k) - \text{loss}(u_{\text{cmd}}(k)), \label{eq:deadzone_loss}
\end{equation}
where $u_{\text{cmd}}(k) \in \mathbb{R}^m$ is the algorithmically commanded control signal. 

Traditionally, compensating for this behavior requires identifying complex, parameterized nonlinear models to explicitly invert the dead-zone mapping. However, this paper proposes a paradigm shift: our observer does not need to know the shape or bounds of $\text{loss}(\cdot)$. Instead, we lump this nonlinear deficit into a single, time-varying physical mismatch term, defined as $d(k) \triangleq -\text{loss}(u_{\text{cmd}}(k))$. Consequently, the effective input simplifies to:
\begin{equation}
    u_{\text{eff}}(k) = u_{\text{cmd}}(k) + d(k). \label{eq:eff_input}
\end{equation}

\begin{assumption}[Slowly Varying Actuator Mismatch]\label{as:dc}
    The physical mismatch $d(k)$ varies slowly relative to the observer horizon, with a bounded single-step increment $\|d(k)-d(k-1)\| \le \sigma_d$ for some $\sigma_d \ge 0$. The piecewise-constant case $d(k)\approx d$ is recovered as the special case $\sigma_d = 0$.
\end{assumption}

Because the sampling frequency is high and the observer uses a highly truncated historical window $N_o$, the dead-band loss is well approximated as a locally constant DC offset acting on the linear plant over the estimation horizon. The bounded-rate form stated above is the condition actually invoked in the stability analysis of Section~\ref{sec:stability}; the constant approximation $d(k)\approx d$ is used only to build the observer regressor in Section~\ref{sec:projection}.

\subsection{Main Control Objective}
Given the availability of exclusively offline and online input-output measurement data (without any prior knowledge of $A, B$, or $C$, nor the structure of the dead-zone), the primary objective of this proposed architecture is twofold:
\begin{enumerate}
    \item \textbf{Mismatch Extraction:} To design a data-driven Extended State Observer (DDESO) that extracts the unknown physical mismatch $d(k)$ directly from raw data. By projecting the recent input-output data onto an offline-constructed behavioral matrix, the observer captures the dead-zone drop-off without state reconstruction or a parameterized nonlinear model.
    \item \textbf{Offset-Free Subspace Predictive Control:} To integrate the estimated mismatch proxy $\hat{d}(k)$ into a data-driven Subspace Predictive Control (SPC)~\cite{favoreel1999} formulation. The control law uses this proxy to shift its baseline control effort so that the controller crosses the dead-band boundary, giving offset-free tracking of a reference $r(k)$.
\end{enumerate}

\section{Nominal Subspace Predictors}\label{sec:predictors}

The proposed architecture uniquely requires the simultaneous identification of two distinct prediction models from a single, disturbance-free dataset. First, we require an \textit{absolute} predictor to extract the steady-state mapping of the plant; this is vital for the disturbance observer. Second, we require an \textit{incremental} predictor for the control algorithm; this naturally embeds integral action to guarantee offset-free tracking.

\subsection{Offline Data Collection and Hankel Matrices}
To capture the fundamental dynamics of the unknown plant, we conduct an offline identification experiment. The plant is excited using a persistently exciting input sequence, $u_{\text{spc}}(k) \in \mathbb{R}^m$, over a total duration of $T_d$ samples. The corresponding output sequence, $y_{\text{spc}}(k) \in \mathbb{R}^p$, is recorded. We strictly assume this offline dataset is collected in a nominal, disturbance-free environment where $d \equiv 0$.

We define a past historical horizon, $T_{\text{ini}}$, and a future prediction horizon, $N$. The parameter $T_{\text{ini}}$ must be chosen such that it is greater than or equal to the observability index of the underlying system. Let $M = T_d - T_{\text{ini}} - N + 1$ denote the total number of block columns that can be extracted from the dataset. 
For an arbitrary vector sequence $v(k)$, we define the block-Hankel matrix of depth $\ell$ as:
\begin{equation*}
    \mathcal{H}_{\ell}(v) = 
    \begin{bmatrix}
        v(0) & v(1) & \cdots & v(M-1) \\
        v(1) & v(2) & \cdots & v(M) \\
        \vdots & \vdots & \ddots & \vdots \\
        v(\ell-1) & v(\ell) & \cdots & v(\ell+M-2)
    \end{bmatrix}.
\end{equation*}

\subsection{The Absolute Predictor: Capturing Steady-State Mappings}
The absolute predictor models the direct linear relationship between the raw magnitudes of the inputs and outputs. Because this formulation inherently lacks integral action, its steady-state predictions are sensitive to input offsets. We will mathematically exploit this sensitivity in Section~\ref{sec:proxy} to isolate the physical mismatch.

We construct the absolute past input, past output, and future output Hankel matrices directly from the raw experimental data:
\begin{align*}
    U_p^a &= \mathcal{H}_{T_{\text{ini}}}(u_{\text{spc}}) \in \mathbb{R}^{m T_{\text{ini}} \times M}, \\
    Y_p^a &= \mathcal{H}_{T_{\text{ini}}}(y_{\text{spc}}) \in \mathbb{R}^{p T_{\text{ini}} \times M}, \\
    Y_f^a &= \mathcal{H}_{N}(y_{\text{spc}}) \in \mathbb{R}^{p N \times M}.
\end{align*}

According to Willems' Fundamental Lemma~\cite{willems2005}, if the input sequence is persistently exciting of a sufficient order, the column space of these Hankel matrices spans all valid finite trajectories of the LTI system. Therefore, there exist linear mapping matrices $P_{1,a} \in \mathbb{R}^{pN \times m T_{\text{ini}}}$ and $P_{2,a} \in \mathbb{R}^{pN \times p T_{\text{ini}}}$ such that the future absolute outputs can be predicted by the past absolute trajectories:
\begin{equation}
    Y_f^a = P_{1,a} U_p^a + P_{2,a} Y_p^a. \label{eq:abs_pred_matrix}
\end{equation}

To identify these matrices from experimental data subject to measurement noise, we formulate a regularized least-squares optimization problem. We minimize the Frobenius norm of the prediction residual, supplemented by a Tikhonov regularization term $\lambda > 0$ to ensure numerical well-conditioning:
\begin{equation*}
    \min_{P_{1,a}, P_{2,a}} \left\| Y_f^a - \begin{bmatrix} P_{1,a} & P_{2,a} \end{bmatrix} \begin{bmatrix} U_p^a \\ Y_p^a \end{bmatrix} \right\|_F^2 + \lambda \left\| \begin{bmatrix} P_{1,a} & P_{2,a} \end{bmatrix} \right\|_F^2.
\end{equation*}
The closed-form analytical solution is given by:
\begin{equation}
    \begin{bmatrix} P_{1,a} & P_{2,a} \end{bmatrix} = Y_f^a \begin{bmatrix} U_p^a \\ Y_p^a \end{bmatrix}^\top \left( \begin{bmatrix} U_p^a \\ Y_p^a \end{bmatrix} \begin{bmatrix} U_p^a \\ Y_p^a \end{bmatrix}^\top + \lambda I \right)^{-1}. \label{eq:abs_least_squares}
\end{equation}

\subsection{The Incremental Predictor: Embedding Integral Action}
While the absolute predictor is essential for observer design, employing it directly within a predictive controller would result in steady-state tracking errors if any unmodeled dynamics exist. To resolve this, we derive an incremental predictor for the online control algorithm.

We define the backward difference operator for the input sequence as $\Delta u_{\text{spc}}(k) = u_{\text{spc}}(k) - u_{\text{spc}}(k-1)$. Using this differenced data, we construct the incremental past and future input Hankel matrices. Crucially, we maintain the \textit{absolute} coordinates for the output matrices to ensure the controller tracks the true reference magnitude, rather than just output variations:
\begin{align*}
    \Delta U_p &= \mathcal{H}_{T_{\text{ini}}}(\Delta u_{\text{spc}}) \in \mathbb{R}^{m T_{\text{ini}} \times M}, \\
    \Delta U_f &= \mathcal{H}_{N}(\Delta u_{\text{spc}}) \in \mathbb{R}^{m N \times M}.
\end{align*}
   
The future absolute outputs are then modeled as a linear combination of the past input increments, past absolute outputs, and future planned input increments:
\begin{equation}
    Y_f^a = P_1 \Delta U_p + P_2 Y_p^a + \Gamma \Delta U_f, \label{eq:inc_pred_matrix}
\end{equation}
where $P_1 \in \mathbb{R}^{pN \times m T_{\text{ini}}}$, $P_2 \in \mathbb{R}^{pN \times p T_{\text{ini}}}$, and $\Gamma \in \mathbb{R}^{pN \times m N}$. Similar to the absolute case, these predictor matrices are identified via a regularized right-inverse projection onto the combined data space:
\begin{equation}
    \begin{bmatrix} P_1 & P_2 & \Gamma \end{bmatrix} = Y_f^a \begin{bmatrix} \Delta U_p \\ Y_p^a \\ \Delta U_f \end{bmatrix}^\top \left( \begin{bmatrix} \Delta U_p \\ Y_p^a \\ \Delta U_f \end{bmatrix} \begin{bmatrix} \Delta U_p \\ Y_p^a \\ \Delta U_f \end{bmatrix}^\top + \lambda I \right)^{-1}.
\end{equation}

This offline identification yields the final nominal incremental prediction model utilized for online control execution:
\begin{equation}
    \mathbf{y}_f = P_1 \Delta \mathbf{u}_{\text{cmd},p} + P_2 \mathbf{y}_p + \Gamma \Delta \mathbf{u}_{\text{cmd},f}. \label{eq:inc_online_model}
\end{equation}
By formulating the model such that the control variable is the input increment ($\Delta \mathbf{u}_{\text{cmd},f}$), the steady-state control objective inherently requires $\Delta u \to 0$ as $y \to r$. Consequently, this velocity-form structure naturally embeds an integrator into the closed-loop system, facilitating robust offset-free tracking.

\section{Absolute Prediction Error as a Mismatch Proxy}\label{sec:proxy}

A limitation of using an incremental predictor for state or disturbance estimation is its integral action. If a constant physical mismatch, such as a dead-zone boundary, acts on the plant, the incremental variations ($\Delta u, \Delta y$) decay to zero in steady state. The incremental prediction error then vanishes, removing any steady-state information about the magnitude of the mismatch.

To circumvent this, we utilize the absolute predictor identified in Section~\ref{sec:predictors}. Because the absolute predictor maps the raw magnitudes of the inputs and outputs without embedded integrators, any persistent unmodeled mismatch will manifest as a proportional, non-zero steady-state prediction error. This section rigorously derives the algebraic mapping from this geometric residual to a purely data-driven mismatch proxy.
The geometric distinction between these two tracking architectures and their divergent capability to preserve steady-state signature bounds is illustrated in Fig. \ref{fig:predictors}.

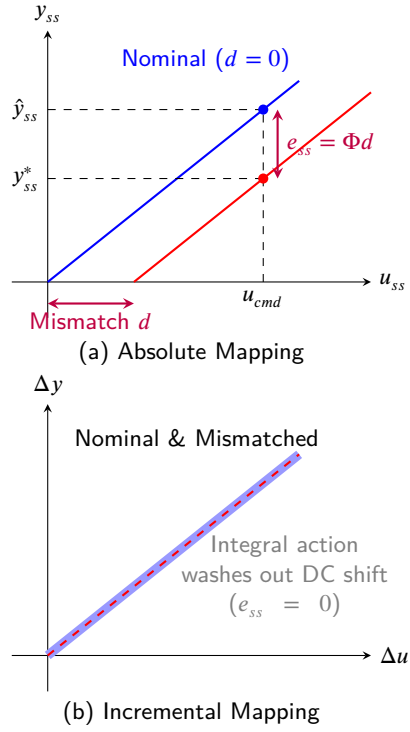
\begin{figure}[htbp]
    \centering
    \begin{tikzpicture}[scale=0.95, >=stealth]
        \begin{scope}[shift={(0,5.2)}]
            \draw[->] (-0.5, 0) -- (4.5, 0) node[right] {$u_{ss}$};
            \draw[->] (0, -0.5) -- (0, 3.5) node[above] {$y_{ss}$};
            \node[font=\small] at (2.0, -0.99) {(a) Absolute Mapping};
            
            \draw[thick, blue] (0,0) -- (3.5, 2.8) node[above left, font=\small] {Nominal ($d=0$)};
            
            \draw[thick, red] (1.2,0) -- (4.5, 2.64);
            
            \coordinate (U) at (3.0, 0);
            \coordinate (Ynom) at (3.0, 2.4); 
            \coordinate (Ymis) at (3.0, 1.44); 
            
            \draw[dashed] (U) node[below, font=\small] {$u_{cmd}$} -- (Ynom);
            \draw[dashed] (0, 2.4) node[left, font=\small] {$\hat{y}_{ss}$} -- (Ynom);
            \draw[dashed] (0, 1.44) node[left, font=\small] {$y^*_{ss}$} -- (Ymis);
            
            \fill[blue] (Ynom) circle (2pt);
            \fill[red] (Ymis) circle (2pt);
            
            \draw[<->, thick, purple] (3.2, 1.44) -- (3.2, 2.4) node[midway, right, font=\small] {$e_{ss} = \Phi d$};
            \draw[<->, thick, purple] (0, -0.3) -- (1.2, -0.3) node[midway, below, font=\small] {Mismatch $d$};
        \end{scope}

        \begin{scope}[shift={(0,0)}]
            \draw[->] (-0.5, 0) -- (4.5, 0) node[right] {$\Delta u$};
            \draw[->] (0, -0.5) -- (0, 3.5) node[above] {$\Delta y$};
            \node[font=\small] at (2.0, -0.8) {(b) Incremental Mapping};
            
            \draw[line width=3.5pt, blue!40] (0,0) -- (3.5, 2.8);
            \draw[thick, red, dashed] (0,0) -- (3.5, 2.8);
            \node[above left, font=\small] at (3.9, 2.8) {Nominal \& Mismatched};
            
            \node[align=center, text width=4.0cm, gray, font=\small] at (3.3, 1.1) {Integral action \\ washes out DC shift \\ ($e_{ss} = 0$)};
        \end{scope}
    \end{tikzpicture}
    \caption{Geometric comparison of steady-state mappings. (a) The absolute predictor captures the dead-zone as a non-zero prediction residual $e_{ss}$ proportional to the mismatch $d$. (b) The incremental predictor's variables ($\Delta u, \Delta y$) wash out the constant offset, causing the nominal and mismatched trajectory spaces to overlap and destroying the dead-zone signature.}
    \label{fig:predictors}
\end{figure}

\subsection{Defining the Absolute Prediction Error}
Let $k$ denote an arbitrary sampling instant during a secondary operational phase. We define the immediate historical input and output trajectories over the past horizon $T_{\text{ini}}$ as:
\begin{align*}
    \mathbf{u}_{\text{past}}(k) &= \col\left( u(k-T_{\text{ini}}), \dots, u(k-1) \right) \in \mathbb{R}^{m T_{\text{ini}}}, \\
    \mathbf{y}_{\text{past}}(k) &= \col\left( y(k-T_{\text{ini}}), \dots, y(k-1) \right) \in \mathbb{R}^{p T_{\text{ini}}}.
\end{align*}
Using the offline-identified absolute mapping matrices from \eqref{eq:abs_pred_matrix}, the one-step-ahead nominal prediction of the output at time $k$ is computed as:
\begin{equation*}
    \hat{y}(k) = P_{1,a}^{(1)} \mathbf{u}_{\text{past}}(k) + P_{2,a}^{(1)} \mathbf{y}_{\text{past}}(k) \in \mathbb{R}^p,
\end{equation*}
where $P_{1,a}^{(1)} \in \mathbb{R}^{p \times m T_{\text{ini}}}$ and $P_{2,a}^{(1)} \in \mathbb{R}^{p \times p T_{\text{ini}}}$ represent the first block rows of $P_{1,a}$ and $P_{2,a}$, respectively, corresponding to a single future prediction step. 

The absolute prediction error is subsequently defined as the geometric residual between the true measured output and this nominal, mismatch-ignorant prediction:
\begin{equation}
    e_{\text{spc}}(k) = y(k) - \hat{y}(k) \in \mathbb{R}^p. \label{eq:espc_def}
\end{equation}

\subsection{Steady-State Mapping of the Physical Mismatch}
To prove that $e_{\text{spc}}(k)$ serves as an exact proxy for the unmeasurable physical mismatch $d$, we analyze the system's behavior at steady state. Let the true, unknown steady-state DC gain of the physical LTI plant be denoted by $G_0 \in \mathbb{R}^{p \times m}$. 

First, consider the nominal system operating in a purely linear regime without mismatch ($d=0$). Under a constant steady-state input $u_{ss}$, the plant reaches a steady-state output $y_{ss} = G_0 u_{ss}$. Because the absolute predictor characterizes the nominal plant in the noise-free case, the one-step prediction matches the true output ($\hat{y}_{ss} = y_{ss}$). Substituting these steady-state vectors into the predictor yields:
\begin{equation}
    y_{ss} = P_{1,a}^{(1)} (\mathbf{1}_{T_{\text{ini}}} \otimes I_m) u_{ss} + P_{2,a}^{(1)} (\mathbf{1}_{T_{\text{ini}}} \otimes I_p) y_{ss}.
\end{equation}
Rearranging this identity in terms of the unknown physical DC gain $G_0$ establishes a fundamental behavioral relationship:
\begin{equation}
    \left[ I_p - P_{2,a}^{(1)} (\mathbf{1}_{T_{\text{ini}}} \otimes I_p) \right] G_0 = P_{1,a}^{(1)} (\mathbf{1}_{T_{\text{ini}}} \otimes I_m). \label{eq:nominal_identity}
\end{equation}

Now, suppose the actuator crosses into a nonlinear dead-band, subjecting the plant to an unknown, piecewise-constant physical mismatch $d \neq 0$. As established in \eqref{eq:eff_input}, the new effective input becomes $u_{\text{eff}} = u_{ss} + d$, which drives the true physical plant to a shifted steady state:
\begin{equation}
    y_{ss}^* = G_0 (u_{ss} + d).
\end{equation}
The absolute predictor, remaining ignorant of $d$, generates its prediction based strictly on the measured algorithmic commands and observed outputs:
\begin{equation}
    \hat{y}_{ss}^* = P_{1,a}^{(1)} (\mathbf{1}_{T_{\text{ini}}} \otimes I_m) u_{ss} + P_{2,a}^{(1)} (\mathbf{1}_{T_{\text{ini}}} \otimes I_p) y_{ss}^*.
\end{equation}
The steady-state absolute prediction error, $e_{ss} = y_{ss}^* - \hat{y}_{ss}^*$, can now be expanded. Substituting the predictor equation yields:
\begin{align}
    e_{ss} &= y_{ss}^* - \left[ P_{1,a}^{(1)} (\mathbf{1}_{T_{\text{ini}}} \otimes I_m) u_{ss} + P_{2,a}^{(1)} (\mathbf{1}_{T_{\text{ini}}} \otimes I_p) y_{ss}^* \right] \nonumber \\
           &= \left[ I_p - P_{2,a}^{(1)} (\mathbf{1}_{T_{\text{ini}}} \otimes I_p) \right] y_{ss}^* - P_{1,a}^{(1)} (\mathbf{1}_{T_{\text{ini}}} \otimes I_m) u_{ss}.
\end{align}
Substituting the true physical steady-state $y_{ss}^* = G_0 (u_{ss} + d)$ into the residual:
\begin{equation}
    e_{ss} = \left[ I_p - P_{2,a}^{(1)} (\mathbf{1}_{T_{\text{ini}}} \otimes I_p) \right] G_0 (u_{ss} + d) - P_{1,a}^{(1)} (\mathbf{1}_{T_{\text{ini}}} \otimes I_m) u_{ss}.
\end{equation}
Applying the nominal behavioral identity derived in \eqref{eq:nominal_identity} simplifies this expression:
\begin{align}
    e_{ss} &= \left[ P_{1,a}^{(1)} (\mathbf{1}_{T_{\text{ini}}} \otimes I_m) \right] (u_{ss} + d) - P_{1,a}^{(1)} (\mathbf{1}_{T_{\text{ini}}} \otimes I_m) u_{ss} = \left[ P_{1,a}^{(1)} (\mathbf{1}_{T_{\text{ini}}} \otimes I_m) \right] d. \label{eq:err_proof}
\end{align}

Equation~\eqref{eq:err_proof} is the key result of this section. It shows that the steady-state mapping from the physical mismatch to the absolute prediction error is decoupled from the recursive output feedback and is determined solely by the input-to-output block of the absolute predictor.

We formally define this data-driven DC mapping matrix as $\Phi \in \mathbb{R}^{p \times m}$:
\begin{equation}
    \Phi = P_{1,a}^{(1)} (\mathbf{1}_{T_{\text{ini}}} \otimes I_m). \label{eq:phi_def}
\end{equation}
Assuming the system has at least as many sensors as actuators ($p \ge m$) and $\Phi$ is full column rank, we can invert this mapping using the Moore-Penrose pseudo-inverse, $\Phi^{\dagger} \in \mathbb{R}^{m \times p}$. This inversion provides a purely data-driven, instantaneous proxy for the physical actuator mismatch:
\begin{equation}
    d_{\text{proxy}}(k) = \Phi^{\dagger} e_{\text{spc}}(k) \in \mathbb{R}^m. \label{eq:dproxy}
\end{equation}
When $p < m$, as in the two-input, single-output amplifier of Section~\ref{sec:experiments}, $\Phi$ has a nontrivial null space and $\Phi^{\dagger}$ recovers only the component of the mismatch that is observable from the output, namely the projection of $d$ onto the row space of $\Phi$. This is sufficient whenever the dead-band is excited on a single channel at a time (the operating regime considered here); recovering a fully simultaneous multi-channel mismatch would require additional independent outputs.

\subsection{Embedding the Proxy into the Behavioral Hankel Matrix}
To construct the projection framework for the Data-Driven Extended State Observer (DDESO), we compile a secondary operational dataset, denoted by the subscript $\text{tr}$. Unlike the nominal dataset, this trajectory inherently contains real-world physical mismatches (e.g., natural traversals of the dead-zone boundary) across $T_{\text{tr}}$ samples. Crucially, it requires no artificially injected disturbance signals. 

For each discrete time step $k \in [T_{\text{ini}} + 1, T_{\text{tr}}]$, we compute the prediction error $e_{\text{spc}}(k)$ via \eqref{eq:espc_def} and map it to the mismatch proxy $d_{\text{proxy}}(k)$ via \eqref{eq:dproxy}. 

Let $N_o$ denote the chosen observer horizon, which dictates the dynamic memory of the DDESO. We extract the necessary block-Hankel matrices from this proxy-mapped dataset. Let $M_{\text{eso}} = T_{\text{tr}} - T_{\text{ini}} - N_o + 1$ define the number of available data columns. The respective Hankel matrices for the historical inputs, the disturbance proxy, and the true measured outputs are defined as:
\begin{align*}
    U_p^{\text{eso}} &= \mathcal{H}_{N_o}(u_{\text{tr}}) \in \mathbb{R}^{mN_o \times M_{\text{eso}}}, \\
    D_p &= \mathcal{H}_{N_o}(d_{\text{proxy}}) \in \mathbb{R}^{mN_o \times M_{\text{eso}}}, \\
    Y_p^{\text{eso}} &= \mathcal{H}_{N_o}(y_{\text{tr}}) \in \mathbb{R}^{pN_o \times M_{\text{eso}}}.
\end{align*}

Finally, these matrices are vertically stacked to form the behavioral data matrix $H$, which spans the trajectory space of the mismatched system:
\begin{equation}
    H = \begin{bmatrix} U_p^{\text{eso}} \\ D_p \\ Y_p^{\text{eso}} \end{bmatrix} \in \mathbb{R}^{(2mN_o + pN_o) \times M_{\text{eso}}}. \label{eq:H_matrix}
\end{equation}
By explicitly embedding $D_p$ into the central block of $H$, the underlying physical relationship between historical measurement sequences and the unmeasurable dead-zone mismatch is permanently encoded into the column space of the data matrix, setting the stage for instantaneous real-time algebraic projection.

\section{Online Mismatch Estimation via Geometric Projection}\label{sec:projection}

With the comprehensive behavioral data matrix $H$ constructed offline, the control architecture transitions to the real-time operational phase. The objective is to utilize the permanently embedded geometric relationships within $H$ to estimate the instantaneous physical mismatch, $d(k)$. Crucially, this must be achieved using exclusively the most recent rolling window of measurable inputs and outputs.

\subsection{Mathematical Derivation of the Observer Matrix $L_d$}
At any real-time sampling instant $k$, the controller accesses the immediate historical horizon of $N_o$ actual commands and measured outputs. We define these online measurement vectors as:
\begin{align*}
    u_{\text{ini}}(k) &= \col\bigl( u_{\text{cmd}}(k-N_o+1), \dots, u_{\text{cmd}}(k) \bigr) \in \mathbb{R}^{mN_o}, \\
    y_{\text{ini}}(k) &= \col\bigl( y(k-N_o+1), \dots, y(k) \bigr) \in \mathbb{R}^{pN_o}.
\end{align*}

Under Assumption~\ref{as:dc}, the physical mismatch acting upon the actuator is treated as a local constant, $d(k) \approx d$, over this highly truncated memory horizon $N_o$. Therefore, the theoretical trajectory of the mismatch across this window is given by $\mathbf{1}_{N_o} \otimes d \in \mathbb{R}^{mN_o}$. 

\begin{remark}[Observer Horizon Trade-off]
    The choice of the observer horizon length $N_o$ dictates a fundamental trade-off in the DDESO design. If $N_o$ is chosen too small, the geometric projection becomes susceptible to measurement noise, which can lead to erratic mismatch estimates. Conversely, if $N_o$ is too large, the assumption that the physical mismatch $d(k)$ remains approximately constant (Assumption~\ref{as:dc}) may break down, particularly during fast reference transients where the actuator rapidly traverses the nonlinear dead-band. As an illustrative example, for the high-precision power amplifier operating at $f_s = 400$\,kHz, selecting $N_o = 12$ corresponds to a rolling window of $N_o/f_s = 30\,\mu$s (spanning $N_o-1 = 11$ sampling intervals, i.e.\ $27.5\,\mu$s of elapsed time). This short horizon supports the locally constant approximation while retaining enough data depth to attenuate high-frequency sensor noise.
\end{remark}

If this specific mismatch $d$ is currently acting upon the system, the corresponding $N_o$-step trajectory formed by the inputs, the mismatch, and the outputs must be structurally compatible with the physical plant. According to Willems' Fundamental Lemma, this composite vector $w(d)$ must reside entirely within the column space of the offline Hankel matrix $H$:
\begin{equation*}
    w(d) = \begin{bmatrix} u_{\text{ini}}(k) \\ \mathbf{1}_{N_o} \otimes d \\ y_{\text{ini}}(k) \end{bmatrix} \in \operatorname{col}(H).
\end{equation*}

This structural containment within the behavioral subspace is interpreted geometrically as an orthogonal projection problem in Fig.~\ref{fig:projection}.

\begin{figure}[htbp]
    \centering
    \begin{tikzpicture}[scale=1.1, >=stealth]
        \colorlet{plane}{blue!10}
        \colorlet{planeline}{blue!60}
        
        \coordinate (A) at (-2.2, -1.2);
        \coordinate (B) at (2.2, -1.2);
        \coordinate (C) at (3.5, 1.5);
        \coordinate (D) at (-0.9, 1.5);
        
        \filldraw[fill=plane, draw=planeline, thick] (A) -- (B) -- (C) -- (D) -- cycle;
        \node[planeline, right] at (C) {$\text{col}(H)$};
        
        \coordinate (O) at (0.8, -0.3);
        \fill[black] (O) circle (1.5pt) node[below right] {Origin};
        
        \coordinate (W0) at (-0.8, 2.8);
        \draw[->, thick, red] (O) -- (W0) node[above left] {$w_0$};
        
        \coordinate (W0_proj) at (-0.8, 0.6); 
        \draw[dashed, thick, gray] (W0) -- (W0_proj) node[midway, left] {$P w_0$};
        
        \draw[gray, thin] (-0.8, 0.8) -- (-0.6, 0.8) -- (-0.6, 0.6);
        
        \coordinate (Proj) at (1.6, 0.9);
        
        \draw[->, thick, green!60!black] (W0) -- (Proj) node[midway, above right] {$E\hat{d}(k)$};
        \draw[->, thick, blue] (O) -- (Proj) node[midway, below right] {$w(\hat{d}) = w_0 + E\hat{d}(k)$};
    \end{tikzpicture}
    \caption{Geometric interpretation of the Data-Driven Extended State Observer (DDESO). The measured trajectory $w_0$ lies outside the physical data subspace $\text{col}(H)$. The observer extracts the optimal mismatch proxy $\hat{d}(k)$ such that the composite vector $w(\hat{d})$ is orthogonally projected back onto the valid trajectory space.}
    \label{fig:projection}
\end{figure}
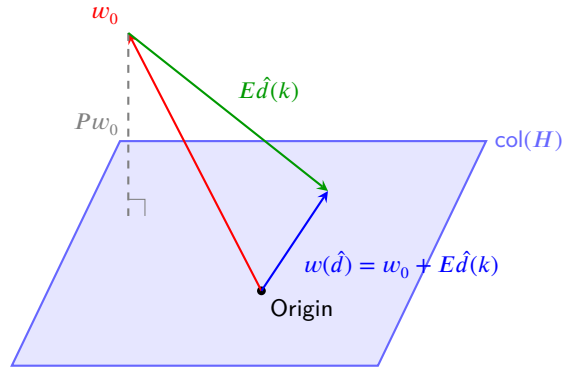

To exploit this geometric property, we define $P$ as the orthogonal projection matrix onto the left null space (the orthogonal complement of the column space) of $H$:
\begin{equation}
    P = I - H H^\dagger \in \mathbb{R}^{(2mN_o + pN_o) \times (2mN_o + pN_o)}. \label{eq:P_proj}
\end{equation}
Because $w(d) \in \operatorname{col}(H)$, the ideal residual of the projection $P w(d)$ must strictly equal zero. To algebraically isolate the unknown variable $d$, we decompose the theoretical composite trajectory into a strictly measurable component, $w_0$, and an unmeasurable mismatch component, parameterized by a selector matrix $E$:
\begin{equation*}
    w(d) = \underbrace{\begin{bmatrix} u_{\text{ini}}(k) \\ \mathbf{0}_{mN_o \times 1} \\ y_{\text{ini}}(k) \end{bmatrix}}_{w_0} 
    + \underbrace{\begin{bmatrix} \mathbf{0}_{mN_o \times m} \\ \mathbf{1}_{N_o} \otimes I_m \\ \mathbf{0}_{pN_o \times m} \end{bmatrix}}_{E} d,
\end{equation*}
where $E \in \mathbb{R}^{(2mN_o + pN_o) \times m}$. The optimal data-driven estimate $\hat{d}(k)$ is formulated as the solution to an unconstrained least-squares optimization problem, which minimizes the squared Euclidean norm of the projection residual:
\begin{equation*}
    \hat{d}(k) = \arg\min_{d} \mathcal{J}(d) = \arg\min_{d} \bigl\| P (w_0 + E d) \bigr\|_2^2.
\end{equation*}

Since $P$ is an orthogonal projector, it is naturally symmetric and idempotent ($P = P^\top = P^2$). Expanding the objective function yields:
\begin{align*}
    \mathcal{J}(d) &= (w_0 + E d)^\top P^\top P (w_0 + E d) \nonumber \\
                   &= (w_0 + E d)^\top P (w_0 + E d) \nonumber \\
                   &= w_0^\top P w_0 + 2 d^\top E^\top P w_0 + d^\top (E^\top P E) d.
\end{align*}
Because the objective function is strictly convex with respect to $d$, the global minimum is located by deriving the gradient with respect to $d$ and equating it to zero:
\begin{equation*}
    \nabla_d \mathcal{J}(d) = 2 E^\top P w_0 + 2 (E^\top P E) d = 0.
\end{equation*}
Solving for $d$ isolates the instantaneous mismatch estimate:
\begin{equation}
    \hat{d}(k) = -\bigl( E^\top P E \bigr)^{-1} E^\top P w_0. \label{eq:d_hat_intermediate}
\end{equation}
We note that $(E^\top P E) \in \mathbb{R}^{m \times m}$ is a highly compact matrix whose inversion is computationally trivial and performed strictly offline. Finally, we map the measurable trajectory $w_0$ directly to the raw, real-time data arrays utilizing a block-diagonal selector matrix $S_{uy} \in \mathbb{R}^{(2mN_o + pN_o) \times (mN_o + pN_o)}$:
\begin{equation*}
    w_0 = S_{uy} \begin{bmatrix} u_{\text{ini}}(k) \\ y_{\text{ini}}(k) \end{bmatrix}, \qquad 
    S_{uy} = \begin{bmatrix} 
        I_{mN_o} & \mathbf{0}_{mN_o \times pN_o} \\ 
        \mathbf{0}_{mN_o \times mN_o} & \mathbf{0}_{mN_o \times pN_o} \\ 
        \mathbf{0}_{pN_o \times mN_o} & I_{pN_o} 
    \end{bmatrix}.
\end{equation*}
Substituting $w_0$ into \eqref{eq:d_hat_intermediate} establishes the final, fixed linear observer mapping:
\begin{equation}
    \hat d(k) = L_d \begin{bmatrix} u_{\text{ini}}(k) \\ y_{\text{ini}}(k) \end{bmatrix}, \qquad
    L_d = -\bigl( E^\top P E \bigr)^{-1} E^\top P S_{uy}. \label{eq:Ld_final}
\end{equation}
The projection matrix $L_d \in \mathbb{R}^{m \times (mN_o + pN_o)}$ is computed entirely offline, bridging the complex trajectory space to a simple runtime gain. 

\subsection{Computational Elegance of Algebraic Projection}
The derivation of $L_d$ highlights the computational advantage of this formulation for real-time implementation.

In traditional observer frameworks, estimating an unmodeled dynamic or disturbance requires an active sequence of operations at every sampling instant. The controller must compute a nominal output prediction using system matrices, calculate the instantaneous tracking error, and filter this error through a dynamic estimator gain via numerical integration. This sequence relies heavily on an accurate parametric model and consumes vital online processing bandwidth.

Conversely, the geometric projection matrix $L_d$ derived in \eqref{eq:Ld_final} fundamentally bypasses this dynamic sequence. Because the absolute predictor's exact DC-mapping was embedded directly into the offline Hankel matrix $H$ (as established in Section~\ref{sec:proxy}), the linear mapping from the historical measurement space to the physical mismatch space is permanently captured within the left null space of $H$. 

Consequently, $L_d$ acts as a direct algebraic extractor. At run time, for example on the hardware-in-the-loop target used here, the algorithm requires no injected test signals, no parameterization of the dead-zone bounds, and no online prediction-error computation. The fixed operation $L_d \begin{bmatrix} u_{\text{ini}}(k) \\ y_{\text{ini}}(k) \end{bmatrix}$ combines prediction, residual generation, and mismatch extraction into a single matrix-vector product.

\section{Recursive Feasibility and Practical Stability Analysis}\label{sec:stability}

In this section, we establish the theoretical guarantees of the proposed integrated architecture. The analysis is divided into two sequential components: demonstrating that the data-driven Subspace Predictive Control (SPC) algorithm remains recursively feasible despite physical mismatches, and proving the practical closed-loop stability (Input-to-State Stability) of the interconnected structure.

\subsection{Definitions and System Setup}
To formalize the analysis, let the augmented historical data vector at time $k$ be denoted as $\chi(k) = \col\bigl( \Delta \mathbf{u}_{\text{cmd},p}(k), \mathbf{y}_p(k) \bigr)$. At each time step, the predictive controller solves the following quadratic program (QP):
\begin{subequations}
\begin{align}
    \mathcal{J}_N^*(\chi(k), &\hat{d}(k)) = \min_{\Delta \mathbf{u}_{\text{cmd},f}} \quad  \bigl\| \mathbf{y}_f - \mathbf{r} \bigr\|_Q^2 + \bigl\| \Delta \mathbf{u}_{\text{cmd},f} \bigr\|_R^2 \label{eq:qp_cost} \\
    \text{s.t.} \quad & \mathbf{y}_f = P_1 \Delta \mathbf{u}_{\text{cmd},p} + P_2 \mathbf{y}_p + \Gamma \Delta \mathbf{u}_{\text{cmd},f} + P_d \hat{d}(k), \label{eq:qp_pred} \\
    & \Delta u_{\min} \le \Delta u_{\text{cmd}}(k+i|k) \le \Delta u_{\max},  \label{eq:qp_rate_con} \\
    & u_{\min} \le u_{\text{cmd}}(k-1) + \sum_{j=0}^{i} \Delta u_{\text{cmd}}(k+j|k) \le u_{\max}, \label{eq:qp_abs_con} \\
    & y(k+N|k) \in \Omega_f. \label{eq:qp_term_con}
\end{align}
\end{subequations}
$\forall i \in [0, N-1]$. Constraint \eqref{eq:qp_term_con} enforces that the terminal predicted output reaches a designated terminal set $\Omega_f$ centered around the target reference $r$. The matrix $P_d \in \mathbb{R}^{pN \times m}$ in~\eqref{eq:qp_pred} propagates the estimated mismatch $\hat{d}(k)$ across the prediction horizon, shifting the predicted trajectory to account for the actuator drop-off; it is obtained directly from the identified predictor blocks.

Instead of relying on heuristic assumptions regarding the observer's accuracy and the terminal set's robustness, we formally derive these properties through the following mathematical lemmas.

\subsection{Theoretical Bounds on the Mismatch Estimation}

\begin{lemma}[Bounded Estimation Error via Orthogonal Projection] \label{lemma:bound_error}
Suppose the physical mismatch $d(k)$ possesses a bounded rate of change, such that $\|d(k) - d(k-1)\| \le \sigma_d$ for some scalar $\sigma_d \ge 0$. Let $\hat{d}(k)$ be the instantaneous mismatch estimate generated by the fixed orthogonal projection $L_d$ derived in \eqref{eq:Ld_final}. Then, the estimation error $\tilde{d}(k) = d(k) - \hat{d}(k)$ is strictly bounded by $\| \tilde{d}(k) \| \le \epsilon_d$, where $\epsilon_d$ scales linearly with $\sigma_d$. Furthermore, if $d(k)$ achieves a constant steady state ($\sigma_d = 0$), then $\tilde{d}(k) \to 0$ once the sliding window is filled with the post-transition value, i.e.\ in at most $N_o$ steps.
\end{lemma}

\begin{proof}
Let the true trajectory of the physical mismatch over the observer horizon $N_o$ be defined as $\mathbf{d}_{N_o}(k) = \col\bigl(d(k-N_o+1), \dots, d(k)\bigr) \in \mathbb{R}^{mN_o}$. We decompose this true trajectory into a theoretical piecewise-constant nominal vector and a variation vector:
\begin{equation}
    \mathbf{d}_{N_o}(k) = \bigl( \mathbf{1}_{N_o} \otimes d(k) \bigr) + \Delta \mathbf{d}_{N_o}(k),
\end{equation}
where $\Delta \mathbf{d}_{N_o}(k)$ represents the accumulated variation within the sliding window. Because the single-step difference is bounded by $\sigma_d$, the maximum norm of the variation vector over the horizon $N_o$ is geometrically bounded: $\|\Delta \mathbf{d}_{N_o}(k)\| \le c_{N_o} \sigma_d$, for some constant $c_{N_o} > 0$ strictly dependent on the horizon length.

The true composite physical trajectory acting on the system is therefore:
\begin{equation}
    w_{\text{true}}(k) = \begin{bmatrix} u_{\text{ini}}(k) \\ \mathbf{1}_{N_o} \otimes d(k) + \Delta \mathbf{d}_{N_o}(k) \\ y_{\text{ini}}(k) \end{bmatrix}.
\end{equation}
By Willems' Fundamental Lemma, any valid trajectory of the LTI system must reside in the column space of the offline Hankel matrix $H$. Thus, $w_{\text{true}}(k) \in \operatorname{col}(H)$. Applying the orthogonal projector $P = I - H H^\dagger$ yields $P w_{\text{true}}(k) = 0$. 

Expanding this projection using the selector matrices $w_0$ and $E$ defined in Section~\ref{sec:projection}:
\begin{align*}
    0 &= P \left( w_0 + E d(k) + S_d \Delta \mathbf{d}_{N_o}(k) \right),
\end{align*}
where $S_d \in \mathbb{R}^{(2mN_o + pN_o) \times mN_o}$ maps the variation vector into the composite space. Rearranging this identity isolates the true current mismatch $d(k)$:
\begin{equation*}
    E^\top P E d(k) = - E^\top P w_0 - E^\top P S_d \Delta \mathbf{d}_{N_o}(k).
\end{equation*}
Pre-multiplying by $-(E^\top P E)^{-1}$ and substituting the observer definition $\hat{d}(k) = -(E^\top P E)^{-1} E^\top P w_0$ yields:
\begin{equation*}
    d(k) = \hat{d}(k) - \bigl( E^\top P E \bigr)^{-1} E^\top P S_d \Delta \mathbf{d}_{N_o}(k).
\end{equation*}
The estimation error is exactly the residual mapping:
\begin{equation}
    \tilde{d}(k) = d(k) - \hat{d}(k) = - \bigl( E^\top P E \bigr)^{-1} E^\top P S_d \Delta \mathbf{d}_{N_o}(k). \label{eq:tilde_d_exact}
\end{equation}
Taking the induced matrix norm provides the explicit strict upper bound:
\begin{equation*}
    \| \tilde{d}(k) \| \le \underbrace{\left\| \bigl( E^\top P E \bigr)^{-1} E^\top P S_d \right\|}_{\triangleq \kappa} c_{N_o} \sigma_d = \epsilon_d.
\end{equation*}
The constant $\kappa$ is entirely defined by offline fixed matrices. Thus, the error is bounded by $\epsilon_d$. If the mismatch becomes constant (e.g., the system settles inside the dead-zone), $\sigma_d = 0 \implies \Delta \mathbf{d}_{N_o}(k) = \mathbf{0}$. From \eqref{eq:tilde_d_exact}, this mandates $\tilde{d}(k) = 0$, proving exact asymptotic estimation.
\end{proof}

\subsection{Robust Terminal Set Construction}

\begin{lemma}[Existence of the Robust Positively Invariant Terminal Set] \label{lemma:robust_set}
Suppose there exists a local linear stabilizing feedback law $\Delta u_{\text{cmd}}(k) = K_f (y(k) - r)$ that stabilizes the nominal prediction model \eqref{eq:qp_pred}. Then, for the bounded estimation error space $\mathcal{W} = \{ P_d \tilde{d} \in \mathbb{R}^{pN} \mid \|\tilde{d}\| \le \epsilon_d \}$, there exists a compact, convex Robust Positively Invariant (RPI) terminal set $\Omega_f$ such that if $y \in \Omega_f$, the system remains in $\Omega_f$ under the local control law despite the continuous perturbation $\tilde{d}(k)$.
\end{lemma}

\begin{proof}
Consider the closed-loop autonomous dynamics of the predicted trajectory operating under the local unconstrained stabilizing feedback law $K_f$. Applying this law to the linear subspace predictor \eqref{eq:qp_pred} yields the closed-loop autonomous transition matrix $\Phi_K$. 

Because $K_f$ is a nominal stabilizing controller, $\Phi_K$ is strictly Schur stable (i.e., its spectral radius $\rho(\Phi_K) < 1$). The evolution of the terminal tracking error, subject to the bounded observer error, is mathematically described by:
\begin{equation*}
    e_y(k+1) = \Phi_K e_y(k) + P_d \tilde{d}(k),
\end{equation*}
where $e_y(k) = y(k) - r$ and the additive perturbation strictly satisfies $P_d \tilde{d}(k) \in \mathcal{W}$.

Because $\mathcal{W}$ is a compact and convex set surrounding the origin, and $\Phi_K$ is Schur stable, standard set-theoretic control literature dictates that the minimal Robust Positively Invariant (mRPI) set is defined by the infinite Minkowski sum of the perturbed reachable spaces:
\begin{equation}
    \Omega_{\infty} = \bigoplus_{i=0}^{\infty} \Phi_K^i \mathcal{W}. \label{eq:minkowski_sum}
\end{equation}
Because $\rho(\Phi_K) < 1$, the infinite Minkowski series \eqref{eq:minkowski_sum} converges to a compact, convex, and bounded set. We define our terminal constraint set as $\Omega_f = \{ y \mid y - r \in \Omega_{\infty} \}$. 

The mechanical mechanism behind this set expansion via Minkowski addition to guard against observer inaccuracies is shown conceptually in Fig. \ref{fig:minkowski}.

\begin{figure}[htbp]
    \centering
    \begin{tikzpicture}[scale=1.7, >=stealth]
        \draw[->, gray] (-2.3, 0) -- (2.3, 0) node[right, font=\small] {$e_{y,1}$};
        \draw[->, gray] (0, -2.0) -- (0, 2.0) node[above, font=\small] {$e_{y,2}$};
        
        \fill[black] (0,0) circle (1.5pt) node[below right, font=\small] {Target $r$};
        
        \draw[dashed, thick, red!60!black] (0,0) ellipse (1.8cm and 1.4cm);
        \node[red!60!black, font=\small] at (1.6, 1.1) {$\Omega_{\text{uncomp}}$};
        
        \draw[thick, blue, fill=blue!10] (0,0) ellipse (1.4cm and 1.0cm);
        \node[blue, font=\small] at (0.5, 1.1) {$\Omega_f$};
        
        \draw[dashed, thin, black] (0,0) ellipse (1.0cm and 0.6cm);
        \node[black, font=\small] at (-0.3,-0.3) {$\Omega_{\text{nom}}$};
        
        \coordinate (P_small) at (-0.5, 0.52);
        \draw[blue, fill=blue!20, opacity=0.7] (P_small) circle (0.4cm);
        \fill[black] (P_small) circle (1pt);
        \draw[->, blue, thick] (P_small) -- ++(-0.28, 0.28) node[midway, below left, font=\scriptsize, inner sep=1pt] {$\mathcal{W}(\epsilon_d)$};
        
        \coordinate (P_large) at (0.866, -0.3);
        \draw[red, fill=red!15, opacity=0.5, dashed] (P_large) circle (0.8cm);
        \fill[black] (P_large) circle (1pt);
        \draw[->, red!80!black, thick] (P_large) -- ++(0.69, -0.4) node[midway, above right, font=\scriptsize, inner sep=1pt] {$\mathcal{W}(d_{\max})$};
        
        \draw[->, thick, orange!80!black, dash dot, rounded corners=8pt] (-1.8, 1.8) -- (-1.8, 1.4) -- (-1.5, 0.9) -- (-1.5, 0.3);
        \fill[orange!80!black] (-1.5, 0.3) circle (1.5pt) node[below, font=\scriptsize, text width=1.5cm, align=center] {Steady-State Offset};
        
        \coordinate (Traj_start) at (1.5, 1.5);
        \coordinate (Mid1) at (1.3, 1.4);
        \coordinate (Mid2) at (0.5, 0.5);
        \draw[->, thick, ForestGreen, rounded corners=12pt] (Traj_start) -- (Mid1) -- (Mid2) -- (0.1, 0.15);
        \node[ForestGreen, font=\small, above right] at (Traj_start) {With DDESO};
        
    \end{tikzpicture}
    \caption{Set-theoretic validation of the observer-driven control loop. Without dead-zone estimation, the controller must absorb the maximal physical deadband profile $\mathcal{W}(d_{\max})$, dictating an immensely conservative robust terminal set $\Omega_{\text{uncomp}}$ plagued by steady-state offsets. Via the proposed algebraic projection, the tracking dynamics shrink down to the nominal space expanded only by the minimal residual estimation boundary $\mathcal{W}(\epsilon_d)$, guaranteeing tight practical convergence into $\Omega_f$.}
    \label{fig:minkowski}
\end{figure}
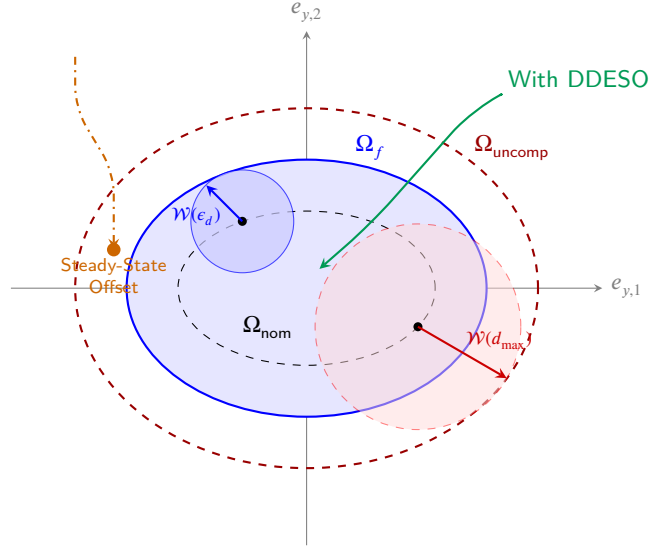

By definition of the Minkowski sum, for any current state $y(k) \in \Omega_f$ and any bounded disturbance $P_d \tilde{d}(k) \in \mathcal{W}$, the subsequent state inherently satisfies:
\begin{equation*}
    y(k+1) - r = \Phi_K (y(k) - r) + P_d \tilde{d}(k) \in \Omega_{\infty}.
\end{equation*}
Thus, $y(k+1) \in \Omega_f$. This proves that a valid, geometrically bounded RPI set $\Omega_f$ exists and inherently absorbs the maximal variation $P_d \epsilon_d$ injected by the estimation error, guaranteeing robust terminal feasibility.
\end{proof}

\subsection{Recursive Feasibility}

\begin{theorem}[Recursive Feasibility] \label{thm:recursive_feasibility}
Suppose the initial optimization problem \eqref{eq:qp_cost} is feasible at time $k=0$ (i.e., a feasible initial solution exists). Under the structural bounds established in Lemma \ref{lemma:bound_error} and the terminal set properties derived in Lemma \ref{lemma:robust_set}, the data-driven predictive controller remains recursively feasible for all $k > 0$, despite the presence of the unmodeled nonlinear physical mismatch.
\end{theorem}

\begin{proof}
Assume the quadratic program \eqref{eq:qp_cost} is feasible at an arbitrary time step $k$, yielding the optimal control increment sequence $\Delta \mathbf{u}_f^*(k)$. According to the receding horizon principle, the first element $\Delta u^*(0|k)$ is executed on the physical plant. At the subsequent time step $k+1$, we construct a sub-optimal candidate control sequence by shifting the previous optimal plan and appending a zero-increment terminal action $\Delta \tilde{\mathbf{u}}_f(k+1) = \bigl\{ \Delta u^*(1|k), \Delta u^*(2|k), \dots, \Delta u^*(N-1|k), \mathbf{0}_{m \times 1} \bigr\}.$
To prove recursive feasibility, we must verify that this candidate sequence strictly satisfies all constraints \eqref{eq:qp_pred}--\eqref{eq:qp_term_con} at time $k+1$.

First, we examine the input magnitude and rate constraints. The sequence elements $\Delta u^*(i|k)$ for $i \in [1, N-1]$ natively satisfy the rate bounds \eqref{eq:qp_rate_con} by virtue of their prior optimality at time $k$. The appended terminal element is the zero vector, which trivially satisfies $\Delta u_{\min} \le \mathbf{0} \le \Delta u_{\max}$. Similarly, the absolute input trajectory induced by this sequence mirrors the previously verified optimal trajectory, with the terminal absolute input remaining statically locked at $u_{\text{cmd}}(k+N|k+1) = u_{\text{cmd}}(k+N-1|k)$. Thus, the absolute constraints \eqref{eq:qp_abs_con} are uniformly satisfied.

Second, we evaluate the terminal state constraint \eqref{eq:qp_term_con}. Due to the physical actuator mismatch, the true initial conditions $\chi(k+1)$ injected into the predictor at $k+1$ have evolved under the true mismatch $d(k)$, while the nominal plan anticipated evolution under the proxy $\hat{d}(k)$. Because the prediction model is strictly linear, the realized terminal state diverges from the planned terminal state by an additive perturbation mapping:
\begin{equation*}
    y(k+N|k+1) = y_{\text{nominal}}(k+N|k) + P_d \tilde{d}(k).
\end{equation*}
From Lemma \ref{lemma:bound_error}, the estimation error is strictly bounded ($\|\tilde{d}(k)\| \le \epsilon_d$). Consequently, the perturbation vector resides within the bounded variation space $\mathcal{W}$. Since $y_{\text{nominal}}(k+N|k) \in \Omega_f$ by the feasibility assumption at $k$, and Lemma \ref{lemma:robust_set} proves that $\Omega_f$ is a Robust Positively Invariant (RPI) set specifically constructed via the Minkowski sum to absorb any perturbation originating from $\mathcal{W}$, we definitively conclude:
\begin{equation*}
    y(k+N|k+1) \in \Omega_f.
\end{equation*}
Thus, the candidate sequence $\Delta \tilde{\mathbf{u}}_f(k+1)$ is strictly feasible, guaranteeing the recursive feasibility of the closed-loop algorithm.
\end{proof}

\subsection{Practical Closed-Loop Stability}
Because the proposed framework estimates a highly nonlinear physical phenomenon (the dead-zone) via a localized linear data-driven proxy, the closed-loop system is formally analyzed through the framework of Input-to-State Stability (ISS).

\begin{theorem}[Practical Stability and Asymptotic Offset-Free Tracking] \label{thm:stability}
Consider the closed-loop system comprising the unknown physical LTI plant \eqref{eq:state}, the offline-derived algebraic projection observer \eqref{eq:Ld_final}, and the Subspace Predictive Controller \eqref{eq:qp_cost}. Under Lemmas \ref{lemma:bound_error} and \ref{lemma:robust_set}, the closed-loop state converges to a bounded positively invariant region around the reference $r$. Moreover, as the transient control increments decay and the effective mismatch settles, the system achieves asymptotic offset-free tracking ($\lim_{k\to\infty} \|y(k) - r\| = 0$).
\end{theorem}

\begin{proof}
We employ the optimal value function $V(k) = \mathcal{J}_N^*(\chi(k), \hat{d}(k))$ as the candidate ISS Lyapunov function. By standard MPC design, $V(k)$ is bounded by class-$\mathcal{K}$ functions: $\alpha_1(\|\chi(k) - \chi_{\text{ref}}\|) \le V(k) \le \alpha_2(\|\chi(k) - \chi_{\text{ref}}\|)$.

To evaluate the Lyapunov descent condition, we assess the optimal cost at time $k+1$ against the upper bound provided by the cost of the feasible, sub-optimal candidate sequence $\Delta \tilde{\mathbf{u}}_f(k+1)$:
\begin{equation*}
    V(k+1) \le \mathcal{J}_N \bigl(\chi(k+1), \hat{d}(k+1), \Delta \tilde{\mathbf{u}}_f(k+1)\bigr).
\end{equation*}
Expanding this quadratic formulation yields the optimal cost from time $k$, penalized by the realized stage cost, and augmented by a trajectory deviation term $\delta(k)$:
\begin{equation}
    V(k+1) \le V(k) - \bigl\| y(k+1|k) - r \bigr\|_Q^2 - \bigl\| \Delta u^*(0|k) \bigr\|_R^2 + \delta(k). \label{eq:lyap_descent}
\end{equation}
The deviation $\delta(k)$ fundamentally arises from the discrepancy between the observer proxy $\hat{d}(k)$ and the true physical mismatch $d(k)$. Based on the condition number of the QP Hessian ($Q, R$) and the linear geometry of the predictor, this deviation is bounded by a local Lipschitz constant $L_{\delta} > 0$:
\begin{equation*}
    \delta(k) \le L_{\delta} \bigl\| P_d \tilde{d}(k) \bigr\| \le L_{\delta} \| P_d \| \epsilon_d.
\end{equation*}
Substituting this worst-case structural bound into \eqref{eq:lyap_descent} yields the ISS descent condition:
\begin{align}
    V(k+1) - V(k) \le& - \bigl\| y(k+1|k) - r \bigr\|_Q^2 - \bigl\| \Delta u^*(0|k) \bigr\|_R^2 + L_{\delta} \| P_d \| \epsilon_d. \label{eq:lyap_final}
\end{align}
Equation \eqref{eq:lyap_final} shows that the Lyapunov function decreases whenever the tracking-error penalty dominates the bounded estimation-error penalty. The state is therefore driven into a bounded neighborhood of the reference, which establishes practical closed-loop stability.

As the state enters this neighborhood and the tracking error decreases, the controller commands increasingly small increments ($\Delta u^*(0|k) \to 0$), so the commanded input $u_{\text{cmd}}$ settles and stops traversing the dead-band boundaries.

The physical mismatch $d(k)$ then stops fluctuating, so its rate of change tends to zero ($\sigma_d \to 0$). By Lemma~\ref{lemma:bound_error}, a constant mismatch drives the null-space estimation error to zero ($\epsilon_d \to 0$); as $\epsilon_d$ vanishes, the perturbation term in \eqref{eq:lyap_final} vanishes ($\delta(k)\to 0$) and asymptotic stability is recovered. The integral action of the velocity-form predictor then removes any residual steady-state deviation, giving offset-free tracking ($\lim_{k\to\infty} y(k) = r$) without a parameterized inversion of the dead-zone.
\end{proof}

\section{Experimental Validation and Case Studies}\label{sec:experiments}

To evaluate the proposed DDESO integrated with subspace predictive control, we report results on two systems. The aim is to assess the recursive feasibility and practical stability established in Section~\ref{sec:stability}, and to demonstrate extraction and rejection of an actuator dead-zone without any \textit{a priori} parametric knowledge of the plant. The framework is first tested in real-time hardware-in-the-loop experiments on a mechanical Multi-DOF Torsion system whose input dead-band is produced by static friction and voltage thresholding. It is then evaluated in a simulation study on a high-precision power amplifier, where the mismatch is a dead-time-induced dead-band in fast-switching power electronics.

\subsection{Case Study 1}
The experimental platform is the Quanser Multi-DOF Torsion system (Fig.~\ref{fig:hardware_setup}), a testbed that reproduces the flexible couplings found in industrial robotics. The plant consists of a Rotary Servo Base Unit coupled through a flexible coupling to a torsion module carrying two inertial discs, giving three rotational inertias (the servo base and the two discs) interconnected by lightly damped torsional springs.

\begin{figure}[htbp]
    \centering
    \includegraphics[width=0.6\columnwidth]{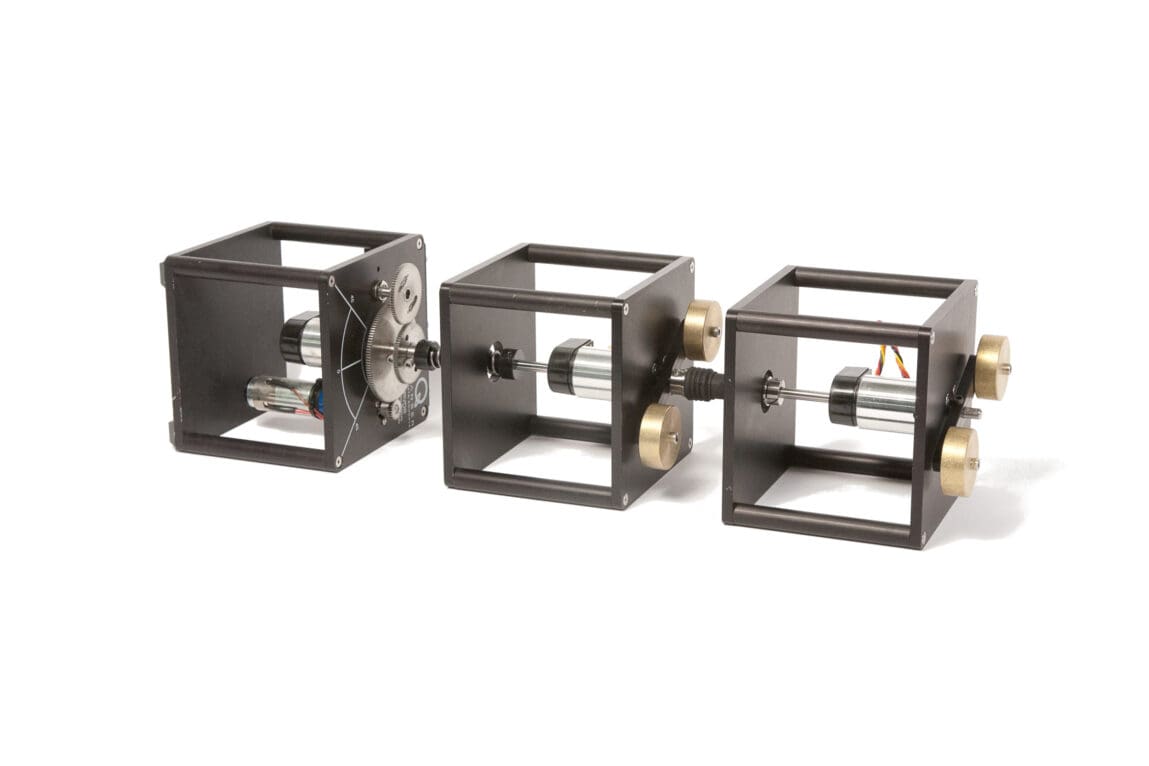}
    \vspace{-2cm}
    \caption{The Quanser Multi-DOF Torsion cyber-physical testbed. The system features a servo base flexibly coupled to multiple inertial loads, presenting a highly oscillatory open-loop response compounded by static Coulomb friction and actuator dead-bands.}
    \label{fig:hardware_setup}
\end{figure}

The modules are connected via a flexible coupling with a nominal torsional stiffness of $1.0\text{ N}\cdot\text{m/rad}$. The rotational dynamics are driven by a direct-current motor and measured by high-resolution optical encoders located on the bearing blocks. Real-time data acquisition and control execution are facilitated by a Quanser Q8-USB device and a VoltPAQ-X1 amplifier, running through the QUARC environment for MATLAB/Simulink.

While the underlying continuous-time dynamics can be nominally approximated by a sixth-order state vector comprising the angular position and velocity of the servo base together with those of the two flexibly coupled torsion discs ($x = [\theta_0, \theta_1, \theta_2, \dot{\theta}_0, \dot{\theta}_1, \dot{\theta}_2]^\top$), the exact system matrices are treated as strictly unknown. The two lightly damped torsional springs render the dominant poles complex-conjugate, producing a lightly damped, oscillatory open-loop response that is difficult to regulate. The actuator also exhibits an unmodeled input dead-zone: static friction in the bearing blocks together with voltage thresholding in the servo drive lump into a measured command dead-band of $\pm 0.18$\,V. For commands $|u_{\text{cmd}}| < 0.18$\,V, the actuator delivers zero effective torque ($u_{\text{eff}} = 0$), which is challenging for a standard linear predictive controller.

Prior to closed-loop execution, an offline identification experiment was conducted. A persistently exciting input sequence $u_{\text{spc}}$ was applied to the nominal system, and the corresponding angular output trajectories $y_{\text{spc}} = \theta_2$ were recorded.

Using this raw dataset, the absolute prediction matrices ($P_{1,a}, P_{2,a}$) and the incremental prediction matrices ($P_1, P_2, \Gamma$) were extracted via regularized least-squares projection as derived in Section~\ref{sec:predictors}. Subsequently, the data-driven DC mapping matrix $\Phi$ was isolated. The secondary proxy-mapped dataset was then compiled, permanently embedding the dead-zone's geometric signature into the central block of the behavioral Hankel matrix $H \in \mathbb{R}^{(2mN_o + pN_o) \times M_{\text{eso}}}$. Finally, the fixed algebraic projection matrix $L_d$ was computed, strictly adhering to the methodology in Section~\ref{sec:projection}.

\subsubsection{Real-Time Control Execution and Results}
During the online phase, the control objective is to track a square-wave reference $r(k)$ in the flexible load's angular position. The controller was tuned with prediction horizon $N = 20$, past-data horizon $T_{\text{ini}} = 15$, observer horizon $N_o = 15$, and weights $Q = 50$, $R = 10^5$; the input and rate constraints were set to the hardware limits of the VoltPAQ-X1 amplifier.

To illustrate the effect of the compensator, the experiment is divided into two phases, shown in Fig.~\ref{fig:exp_results}. In the uncompensated phase ($t \in [0, 20]$\,s, shaded red), the loop runs a standard, mismatch-ignorant subspace predictive controller. The $\pm 0.18$\,V dead-band truncates the control effort near steady state, and because the incremental predictor's integral action removes the steady-state error signal, the controller does not bridge the dead-band, leaving a visible steady-state offset.

\begin{figure}[H]
    \centering
    \includegraphics[width=0.6\columnwidth]{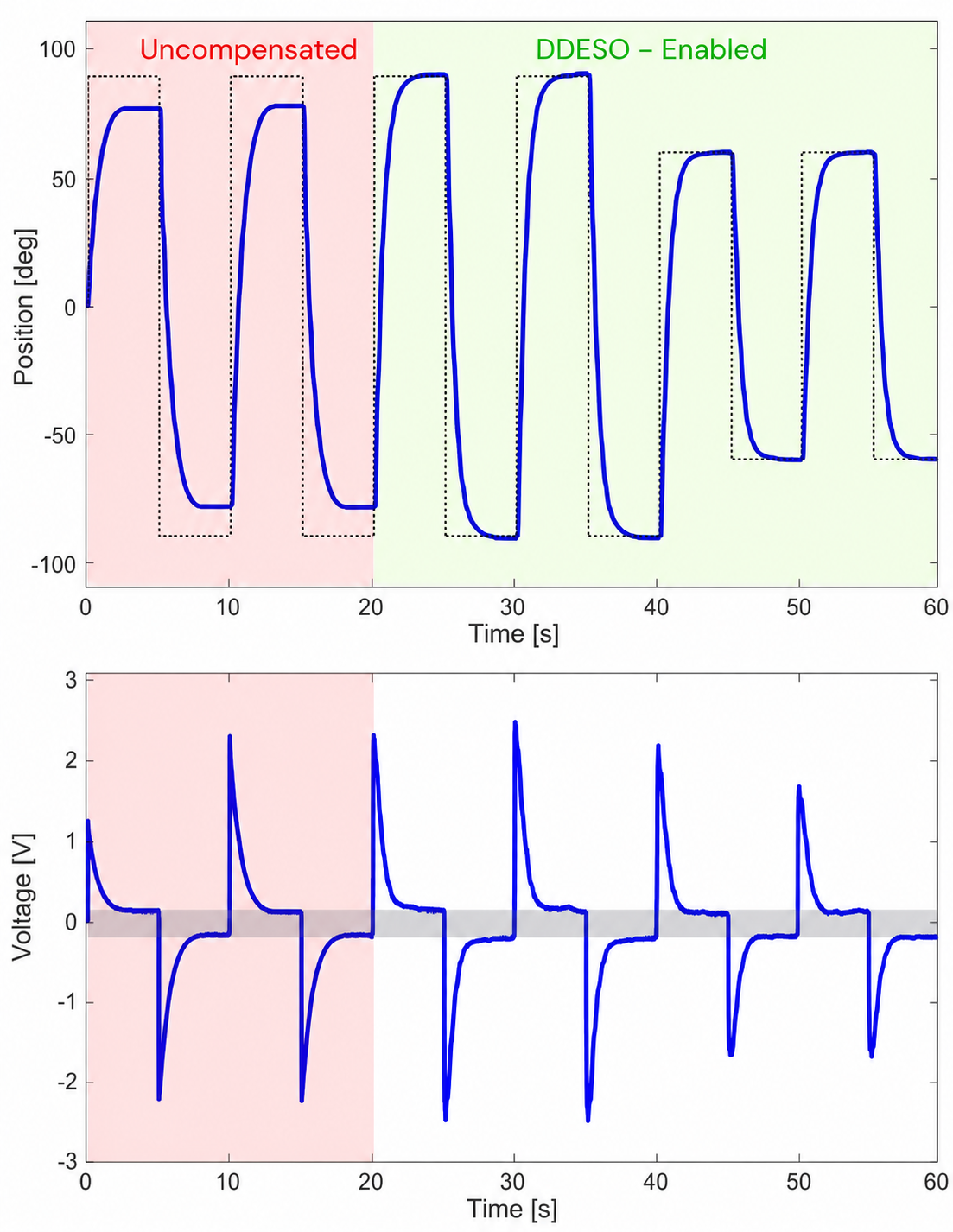}
    \caption{Real-time experimental results on the Quanser Multi-DOF Torsion system. (Top) Tracking performance, showing the steady-state offset during the uncompensated phase ($0$--$20$\,s) and the correction after the proposed DDESO is engaged ($20$--$60$\,s). (Bottom) The control signal $u_{\text{cmd}}$ responding to the estimated dead-zone proxy and crossing the $\pm 0.18$\,V dead-band.}
    \label{fig:exp_results}
\end{figure}

At $t = 20$\,s the proposed DDESO-SPC framework is engaged (shaded green). The fixed projection matrix $L_d$ operates on the current data window of length $N_o$ and detects the drop-off caused by the $\pm 0.18$\,V dead-band. The estimated mismatch is extracted as the proxy $\hat{d}(k)$ and passed to the QP's equality constraint.

Consequently, the predictive controller shifts its baseline control effort. As seen in the control-input trace, the commanded signal $u_{\text{cmd}}$ shifts across the zero-boundary threshold to cross the dead-band. This baseline shift rejects the dead-zone, consistent with Theorem~\ref{thm:stability}, and the integral action then recovers offset-free tracking of the reference without \textit{a priori} parameterization of the nonlinearity.

A supplementary video of the hardware-in-the-loop experiment, including synchronized oscilloscope captures of the states and control inputs, is available online. \footnote{\url{https://www.youtube.com/watch?v=YplHZJuhR2k}}

\subsection{Case Study 2}

To demonstrate the approach on a second system, we consider output-current control of a high-precision power amplifier of the type used in moving stages in the lithography industry. The schematic of this current amplifier is shown in Fig.~\ref{fig:amplifier_schematic}. The system consists of two identical power stages, positive ($p$) and negative ($n$), coupled to a load model.

\begin{figure}[htbp]
    \centering
    \begin{tikzpicture}[scale=1.0, transform shape]
        \draw (0,3) node[circ]{} node[above]{$V_{bus}$}
              to[Tnmos, name=S1] (0,1.5)   
              to[short] (0,1.0);
        \draw (0,1.0) to[Tnmos, name=S2] (0,-0.5) 
              to[short] (0,-1.0);
        \draw (0,-1.0) node[ground]{};
        \node[left] at (-0.75,2.25) {$S_1$};
        \node[left] at (-0.75,0.25) {$S_2$};
        \coordinate (Lmid) at (0,1.0);

        \draw (7,3) node[circ]{} node[above]{$V_{bus}$}
              to[Tnmos, name=S3] (7,1.5)
              to[short] (7,1.0);
        \draw (7,1.0) to[Tnmos, name=S4] (7,-0.5)
              to[short] (7,-1.0);
        \draw (7,-1.0) node[ground]{};
        \node[right] at (7.75,2.25) {$S_3$};
        \node[right] at (7.75,0.25) {$S_4$};
        \coordinate (Rmid) at (7,1.0);

        \draw (Lmid) to[L=$L$, i>^=$i_{Lp}$] (2.2,1.0);
        \draw (7,1.0) to[L=$L$, i<^=$i_{Ln}$] (4.8,1.0);
        \coordinate (Np) at (2.2,1.0);
        \coordinate (Nn) at (4.8,1.0);

        \draw (Np) to[R=$R_m$, i>^=$i_o$] (3.5,1.0)
                    to[L=$L_m$] (Nn);

        \draw (Np) -- (2.2,0.2)
              to[R=$R$] (2.2,-0.6)
              to[C=$C$, v<=$v_{Cp}$] (2.2,-1.4)
              node[ground]{};
        \draw[->,red] (2.55,0.0) -- (2.55,-0.5) node[midway,right,black]{$i_{Cp}$};

        \draw (Nn) -- (4.8,0.2)
              to[R=$R$] (4.8,-0.6)
              to[C=$C$, v<=$v_{Cn}$] (4.8,-1.4)
              node[ground]{};
        \draw[->,red] (5.15,0.0) -- (5.15,-0.5) node[midway,right,black]{$i_{Cn}$};
    \end{tikzpicture}
    \caption{Schematic of the industrial dual-stage current amplifier, adapted from \cite{Lazaroffset}. Each half-bridge leg ($S_1,S_2$ and $S_3,S_4$) drives a series inductor $L$; the output current $i_o$ flows through the load $R_m$--$L_m$, with $R$--$C$ snubber branches at each node.}
    \label{fig:amplifier_schematic}
\end{figure}
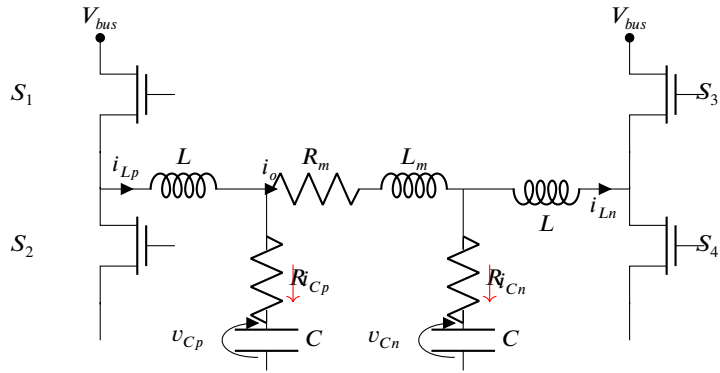

The continuous-time dynamics of this power amplifier are established by applying Kirchhoff's circuit laws across the various switching modes. Let the system state vector be defined as $x = [i_{Lp}\ v_{Cp}\ i_{Ln}\ v_{Cn}\ i_o]^\top$, comprising the inductor currents and capacitor voltages of the power stages alongside the output load current. The measured system output is precisely the output current, $y = i_o$. This yields the following continuous-time state-space representation $\dot{x}(t) = A_c x(t) + B_c u(t),
y(t) = C_c x(t), $ where the continuous system matrices are given by:
\begin{align*}
A_c &= \begin{bmatrix} 
-\frac{R}{L} & -\frac{1}{L} & 0 & 0 & \frac{R}{L} \\ 
\frac{1}{C} & 0 & 0 & 0 & -\frac{1}{C} \\ 
0 & 0 & -\frac{R}{L} & -\frac{1}{L} & -\frac{R}{L} \\ 
0 & 0 & \frac{1}{C} & 0 & \frac{1}{C} \\ 
\frac{R}{L_m} & \frac{1}{L_m} & -\frac{R}{L_m} & -\frac{1}{L_m} & -\frac{2R+R_m}{L_m} 
\end{bmatrix}, \quad 
B_c = \begin{bmatrix} 
\frac{V_{bus}}{L} & 0 \\ 
0 & 0 \\ 
0 & \frac{V_{bus}}{L} \\ 
0 & 0 \\ 
0 & 0 
\end{bmatrix}, \quad
C_c = \begin{bmatrix} 0 \\ 0 \\ 0 \\ 0 \\ 1 \end{bmatrix}^\top.
\end{align*}
where $A_c \in \mathbb{R}^{5 \times 5}$, $B_c \in \mathbb{R}^{5 \times 2}$, and $C_c \in \mathbb{R}^{1 \times 5}$. To facilitate digital control implementation, this linear switched model is discretized utilizing a zero-order-hold equivalent at a sampling frequency $f_s$, resulting in the discrete-time model, in which inputs $u(k) \in \mathbb{R}^2$ represent the duty-cycles applied to the positive and negative power stages. The parameter values defining the physical characteristics of the amplifier are detailed in Table \ref{tab:amp_params}.

\begin{table}[htbp]
\centering
\caption{Industrial Power Amplifier Circuit and Control Parameters}
\label{tab:amp_params}
\begin{tabular}{@{}lll@{}}
\toprule
\textbf{Parameter} & \textbf{Symbol} & \textbf{Value} \\ \midrule
Bus voltage & $V_{bus}$ & $360$\,V \\
Power stage inductance & $L$ & $44$\,$\mu$H \\
Power stage capacitance & $C$ & $0.4$\,$\mu$F \\
Power stage parasitic resistance & $R$ & $62.2$\,$\mu\Omega$ \\
Load model inductance & $L_m$ & $20$\,mH \\
Load model resistance & $R_m$ & $10$\,$\Omega$ \\ \midrule
Sampling Frequency & $f_s$ & $400$\,kHz \\
Prediction Horizon & $N$ & $15$ \\
Past Data Horizon & $T_{ini}$ & $12$ \\
Observer Hankel Window & $N_o$ & $12$ \\
Output Cost Weight & $Q$ & $10$ \\
Input Cost Weight & $R_{cost}$ & $0.1 I_{2}$ \\ \bottomrule
\end{tabular}
\end{table}

\subsubsection{The Dead-Zone Challenge in Power Amplifiers}
The standard incremental Subspace Predictive Control (iSPC) methodology \cite{Lazaroffset} approaches offset-free tracking by formulating the prediction model in terms of input increments, $\Delta u(k) := u(k) - u(k-1)$. By identifying prediction matrices directly from Hankel matrices populated with $\Delta u$, the iSPC natively embeds integral action into the control law.

While the velocity-form iSPC provides offset-free tracking under ideal linear conditions, it is sensitive to the nonlinearities of switching power electronics. In high-precision amplifiers, hardware safety requires a ``dead-time'' (blanking delay) between the gating signals of complementary MOSFETs to prevent shoot-through. This dead-time, together with semiconductor forward voltage drops, produces an actuator dead-zone: within the dead-band the amplifier does not deliver the expected proportional voltage, so the effective input gain drops toward zero.

Because the standard iSPC relies on reactive integration, it does not directly sense the DC nature of this dead-zone. The embedded integrator must accumulate tracking error over many sampling intervals before it commands a $\Delta u$ large enough to cross the dead-band, and this accumulation is further limited by any slew-rate bounds on $\Delta u$. The result is slow transient recovery and residual steady-state offset when operating near the dead-zone boundary.

\begin{figure}[htbp]
    \centering
    \includegraphics[width=0.7\columnwidth]{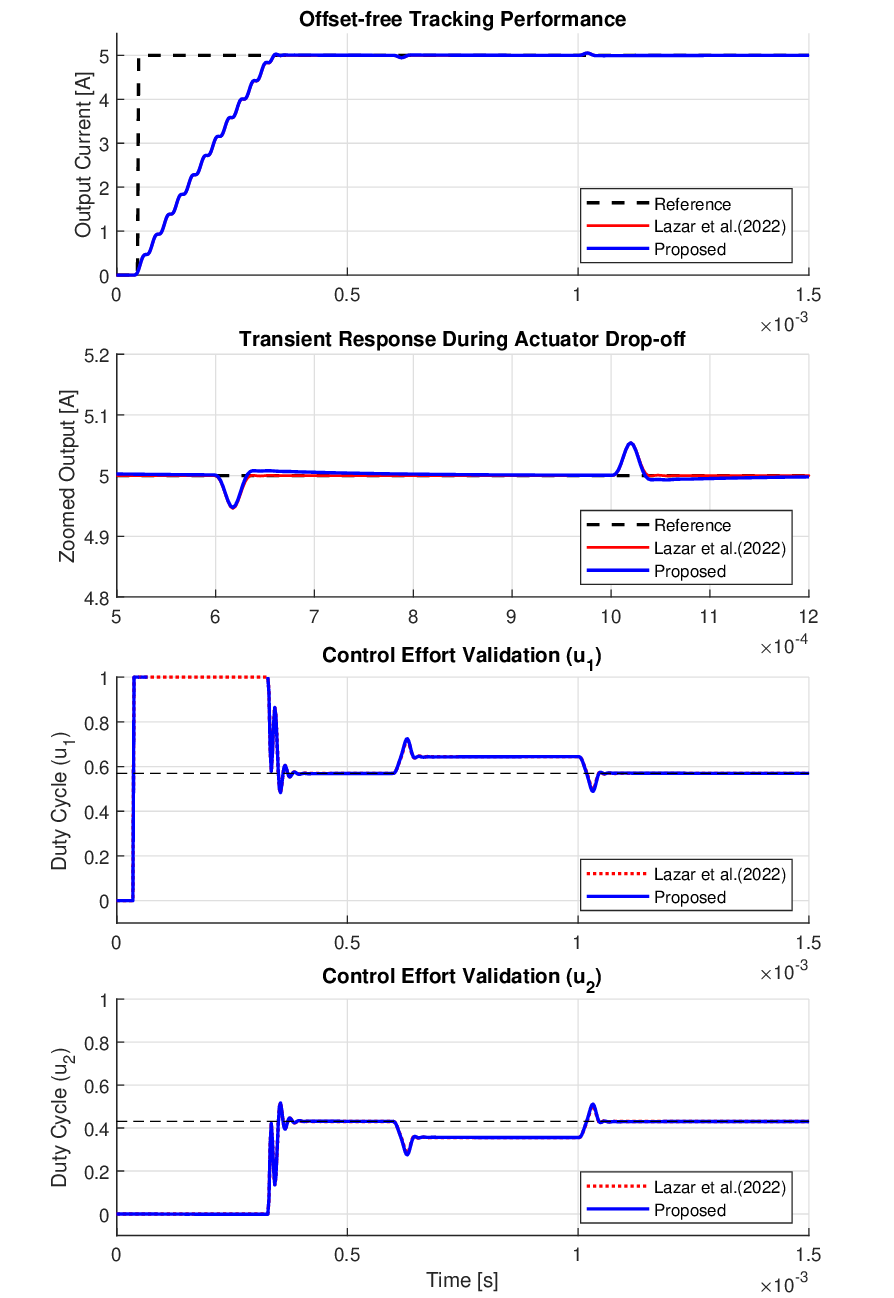}
    \caption{Closed-loop performance of the baseline iSPC \cite{Lazaroffset} and the proposed DDESO-MPC on the high-precision power amplifier (simulation). The system tracks a $5$\,A reference under an unmodeled dead-zone applied between $t = 6 \times 10^{-4}$\,s and $10 \times 10^{-4}$\,s. The DDESO reduces the transient offset seen in the baseline controller by shifting the baseline control effort.}
    \label{fig:amplifier_results}
\end{figure}

\subsubsection{Comparative Results and Technical Discussion}
To evaluate the proposed Data-Driven ESO integrated with MPC, we compare against the baseline iSPC algorithm of Lazar \textit{et al.} \cite{Lazaroffset} under identical tuning and a measurement-noise level of $19$\,dB SNR. Figure~\ref{fig:amplifier_results} shows the tracking performance. To stress the compensator, a dead-zone drop-off is applied to the first input channel ($u_1$) between $t = 6 \times 10^{-4}$\,s and $10 \times 10^{-4}$\,s.

As shown in the zoomed transient (second subplot of Fig.~\ref{fig:amplifier_results}), the baseline iSPC \cite{Lazaroffset} (red dotted trace) exhibits a steady-state offset and an oscillatory recovery, because the velocity-form integration crosses the dead-zone threshold only after accumulating error.

The proposed DDESO-MPC (blue solid trace) reduces this offset. The companion absolute predictor isolates the dead-zone drop-off as a steady-state residual; the orthogonal projection onto the offline behavioral Hankel matrix maps this residual to a proxy of the mismatch, which the controller uses to apply a feedforward shift to its baseline effort rather than waiting for integral accumulation. This is visible in the control-effort subplots, where the duty-cycle commands for the DDESO move away from the baseline iSPC trajectories at the onset of the dead-zone. By shifting across the dead-band, the DDESO recovers offset-free tracking under the tested noise level.

\section{Conclusion}

This paper presented a data-driven architecture for compensating nonlinear actuator mismatches, and dead-zones in particular, in a subspace predictive control setting. We showed that incremental velocity-form predictors remove the steady-state information required for mismatch estimation, and we recovered that information through a companion absolute predictor whose steady-state prediction error serves as a proxy for the physical mismatch. Embedding this DC mapping in a behavioral Hankel matrix yielded a fixed algebraic orthogonal projection, so the controller extracts the mismatch from a short rolling window of raw input-output data with a single matrix-vector product, without a dynamic state estimator.

The analysis established the structural bounds of the framework. Using Minkowski set addition and an Input-to-State Stability (ISS) argument, we showed that the integrated observer-controller loop is recursively feasible and practically stable, and that once the control increments settle the projection recovers the mismatch, giving asymptotic offset-free tracking without prior parameterization of the dead-band or of the underlying linear plant.

These properties were validated in real-time hardware-in-the-loop experiments on a multi-DOF torsion system, where the proposed architecture shifted the subspace predictive controller's baseline effort to cross a $\pm 0.18$\,V physical dead-band and restored offset-free tracking, and in a simulation study on a high-precision power amplifier. Future work will extend the projection framework to rapidly time-varying mismatches and to distributed, multi-agent subspace control, and will study the effect of the observer horizon and noise level on estimation accuracy in more detail.

\bibliographystyle{elsarticle-num}
\bibliography{cas-refs}

@ARTICLE{Liuetal2021,
  author  = {Liu, Y. and Zhang, Y. and Liu, L. and Li, H. and Liu, H.},
  title   = {Adaptive dynamic surface control for a class of dead-zone nonlinear systems via output feedback},
  journal = {IEEE Access},
  year    = {2021}
}

@ARTICLE{Zhangetal2020,
  author  = {Zhang, Y. and Yan, Q. and Cai, J. and Wu, X.},
  title   = {Adaptive iterative learning control for tank gun servo systems with input deadzone},
  journal = {IEEE Access},
  year    = {2020}
}

@ARTICLE{Shouetal2020,
  author  = {Shou, Q. and Xia, H. and Yan, Q. and Cai, J. and Ma, Y. and Bao, Z.},
  title   = {Robust learning control for a class of uncertain linear motor systems},
  journal = {J. Phys.: Conf. Ser.},
  year    = {2020}
}

@ARTICLE{Lietal2021,
  author  = {Li, L. and Lin, Z. and Jiang, Y. and Yu, C. and Yao, J.},
  title   = {Valve deadzone/backlash compensation for lifting motion control of hydraulic manipulators},
  journal = {Machines},
  year    = {2021}
}

@ARTICLE{Wangetal2021,
  author  = {Wang, X. and Ma, Y. and Yan, Q. and Zhang, Y. and Guan, X.},
  title   = {Robust ILC of nonlinearly parametric time-delay systems with input deadzone},
  journal = {J. Phys.: Conf. Ser.},
  year    = {2021}
}

@ARTICLE{Zhangetal2023,
  author  = {Zhang, H. and Yan, Q. and Cai, J. and Gao, S. and Zhang, Y.},
  title   = {Initial-rectification neuro-adaptive iterative learning control for robot manipulators with input deadzone and nonzero initial errors},
  journal = {IEEE Access},
  year    = {2023}
}

@ARTICLE{Yuetal2024,
  author  = {Yu, Y. and Wang, L.},
  title   = {Adaptive fault-tolerant finite-time flight-path angle control for aircraft systems with unknown deadzone and actuator faults},
  journal = {IEEE Access},
  year    = {2024}
}

@ARTICLE{Caoetal2024,
  author  = {Cao, L. and Liu, S. and Xu, L.},
  title   = {An intelligent fault-tolerant control method for a flexible-link manipulator with an uncertain dead-zone and intermittent actuator faults},
  journal = {Mathematics},
  year    = {2024}
}

@MISC{Suetal2018,
  author  = {Su, X. and Liu, Z.},
  title   = {Direct adaptive compensation control of mechanical systems with unknown actuator failures and dead-zone nonlinearities},
  year    = {2018}
}

@MISC{Bessaetal2010,
  author  = {Bessa, W. M. and Dutra, M. and Kreuzer, E.},
  title   = {An adaptive fuzzy dead-zone compensation scheme and its application to electro-hydraulic systems},
  year    = {2010}
}

@UNPUBLISHED{Bessaetal2022,
  author  = {Bessa, W. M. and Dutra, M. and Kreuzer, E.},
  title   = {An adaptive fuzzy dead-zone compensation scheme for nonlinear systems},
  note    = {arXiv preprint},
  year    = {2022}
}

@UNPUBLISHED{Bessa2022,
  author  = {Bessa, W. M.},
  title   = {An adaptive fuzzy sliding mode controller for nonlinear systems with non-symmetric dead-zone and its application to an electro-hydraulic system},
  note    = {arXiv preprint},
  year    = {2022}
}

@ARTICLE{Jungetal2022,
  author  = {Jung, D. and Jeon, J. W.},
  title   = {Synchronous control of 2-D.O.F master-slave manipulators using actuators with asymmetric nonlinear dead-zone characteristics},
  journal = {IEEE Access},
  year    = {2022}
}

@ARTICLE{ChanZhengetal2022,
  author  = {Chan-Zheng, C. and Borja, P. and Scherpen, J.},
  title   = {Dead-zone compensation via passivity-based control for a class of mechanical systems},
  journal = {IFAC-PapersOnLine},
  year    = {2022}
}

@ARTICLE{Fengetal2025,
  author  = {Feng, X. and Wang, C.},
  title   = {Adaptive tracking control for a class of uncertain MIMO nonlinear systems with input constraints},
  journal = {J. Intell. Robot. Syst.},
  year    = {2025}
}

@INPROCEEDINGS{Cequeiraetal2010,
  author    = {de Cequeira, A. C. T. and Vale, M. and Fonseca, D. G. V. and Araújo, F. M. U. and Maitelli, A.},
  title     = {Estimation and compensation of dead-zone inherent to the actuators of industrial processes},
  booktitle = {Int. Conf. Informatics in Control, Automation and Robotics},
  year      = {2010}
}

@ARTICLE{Sunaretal2021,
  author  = {Sunar, N. and others},
  title   = {Improved pole-placement control with feed-forward dead zone compensation for position tracking of electro-pneumatic actuator system},
  journal = {Elektrika},
  year    = {2021}
}

@ARTICLE{Wangetal2023,
  author  = {Wang, T. and others},
  title   = {Active fault-tolerant control for the dual-valve hydraulic system with unknown dead-zone},
  journal = {ISA Trans.},
  year    = {2023}
}

@ARTICLE{Jiaetal2019,
  author  = {Jia, L. and others},
  title   = {A robust adaptive trajectory tracking algorithm using SMC and machine learning for FFSGRs with actuator dead zones},
  journal = {Appl. Sci.},
  year    = {2019}
}

@ARTICLE{Zhouetal2019,
  author  = {Zhou, Q. and Zhao, S. and Li, H. and Lu, R. and Wu, C.},
  title   = {Adaptive neural network tracking control for robotic manipulators with dead zone},
  journal = {IEEE Trans. Neural Netw. Learn. Syst.},
  year    = {2019}
}

@ARTICLE{Maetal2022,
  author  = {Ma, L. and Wang, M.},
  title   = {Adaptive compensation tracking control for time-varying delay nonlinear systems with unknown actuator dead zone},
  journal = {Machines},
  year    = {2022}
}

@ARTICLE{Jang2019,
  author  = {Jang, J.},
  title   = {Fuzzy logic deadzone compensation with feedback linearization of nonlinear systems},
  journal = {Appl. Math.},
  year    = {2019}
}

@MISC{Chiang2013,
  author  = {Chiang, C.-C.},
  title   = {Adaptive fuzzy tracking control for uncertain nonlinear time-delay systems with unknown dead-zone input},
  year    = {2013}
}

@ARTICLE{Tangetal2024,
  author  = {Tang, H. H. and Ahmad, N. S.},
  title   = {Fuzzy logic approach for controlling uncertain and nonlinear systems: a comprehensive review of applications and advances},
  journal = {Syst. Sci. \& Control Eng.},
  year    = {2024}
}

@MISC{Cheongetal2013,
  author  = {Cheong, J. and Han, S. and Lee, J.},
  title   = {Adaptive fuzzy dynamic surface sliding mode position control for a robot manipulator with friction and deadzone},
  year    = {2013}
}

@MISC{Xiangetal2017,
  author  = {Xiang, W. and Li, N. and Sun, Y.},
  title   = {Fuzzy adaptive prescribed performance control for a class of uncertain nonlinear systems with unknown dead-zone inputs},
  year    = {2017}
}

@INPROCEEDINGS{Wangetal2020,
  author    = {Wang, L. and Zhao, D. and Liu, F. and Liu, Q. and Zhang, Z.},
  title     = {Active disturbance rejection position synchronous control of dual-hydraulic actuators with unknown dead-zones},
  booktitle = {Italian National Conf. on Sensors},
  year      = {2020}
}

@ARTICLE{Ahietal2018,
  author  = {Ahi, B. and Haeri, M.},
  title   = {Linear active disturbance rejection control from the practical aspects},
  journal = {IEEE/ASME Trans. Mechatronics},
  year    = {2018}
}

@ARTICLE{Ebrahimietal2024,
  author  = {Ebrahimi, M. M. and Homaeinezhad, M.},
  title   = {Numerical algorithm for nonlinearity compensation of hardly constrained actuation for trajectory tracking control of deadzone-included dynamic systems},
  journal = {ISA Trans.},
  year    = {2024}
}

@ARTICLE{willems2005,
  author  = {Willems, J. C. and Rapisarda, P. and Markovsky, I. and De Moor, B. L. M.},
  title   = {A note on persistency of excitation},
  journal = {Syst. Control Lett.},
  volume  = {54},
  number  = {4},
  pages   = {325--329},
  year    = {2005}
}

@INPROCEEDINGS{coulson2019,
  author    = {Coulson, J. and Lygeros, J. and D\"orfler, F.},
  title     = {Data-enabled predictive control: In the shallows of the DeePC},
  booktitle = {Proc. 18th Eur. Control Conf. (ECC)},
  address   = {Naples, Italy},
  pages     = {307--312},
  year      = {2019}
}

@INPROCEEDINGS{favoreel1999,
  author    = {Favoreel, W. and De Moor, B. and Van Overschee, P.},
  title     = {SPC: Subspace predictive control},
  booktitle = {Proc. 14th IFAC World Congr.},
  address   = {Beijing, China},
  year      = {1999}
}

@ARTICLE{depersis2020,
  author  = {De Persis, C. and Tesi, P.},
  title   = {Formulas for data-driven control: Stabilization, optimality, and robustness},
  journal = {IEEE Trans. Autom. Control},
  volume  = {65},
  number  = {3},
  pages   = {909--924},
  year    = {2020}
}

@INPROCEEDINGS{Lazaroffset,
  author    = {Lazar, M. and Verheijen, P. C. N.},
  title     = {Offset-free data-driven predictive control},
  booktitle = {2022 IEEE 61st conference on decision and control (CDC)},
  pages     = {1099--1104},
  year      = {2022}
}

\end{document}